\documentclass[10pt]{article}
\usepackage[margin=1in]{geometry}  
\pdfoutput=1
\usepackage{graphicx}
\usepackage{amsmath}
\usepackage{amsthm}
\usepackage{algorithm}
\usepackage{pdflscape}
\usepackage{cite}
\usepackage{color}
\usepackage{natbib}
\usepackage{bm}
\usepackage{amssymb}
\usepackage{epstopdf}

\def \Tr{\mbox{Tr}}

\newcommand{\ba}{{\bm a}}

\newcommand{\by}{{\bm y}}
\newcommand{\bY}{{\bm Y}}
\newcommand{\bN}{{\bm N}}
\newcommand{\bZ}{{\bm Z}}
\newcommand{\bP}{{\bm P}}
\newcommand{\bz}{{\bm z}}
\newcommand{\bbeta}{{\bm \beta}}
\newcommand{\bone}{{\bm 1}}
\newcommand{\bzero}{{\bm 0}}
\newcommand{\bI}{{\bm I}}
\newcommand{\bA}{{\bm A}}

\theoremstyle{plain}

\newtheorem{theorem}{Theorem}[section]
\newtheorem{lemma}[theorem]{Lemma}
\newtheorem{corollary}[theorem]{Corollary}

\title{SCALPEL: Extracting Neurons from Calcium Imaging Data}

\begin{document}
\author{Ashley Petersen\footnote{ajpete@uw.edu}, Noah Simon\footnote{nrsimon@uw.edu}, and Daniela Witten\footnote{dwitten@uw.edu} \\ Department of Biostatistics, University of Washington, Seattle WA 98195}
\maketitle


\begin{abstract}
In the past few years, new technologies in the field of neuroscience have made it possible to simultaneously image activity in large populations of neurons at cellular resolution in behaving animals. In mid-2016, a huge repository of this so-called ``calcium imaging" data was made publicly-available. The availability of this large-scale data resource opens the door to a host of scientific questions, for which new statistical methods must be developed.
  
  In this paper, we consider the first step in the analysis of  calcium imaging data: namely, identifying the neurons in a calcium imaging video. We propose a dictionary learning approach for this task. First, we  perform image segmentation to develop a dictionary containing a huge number of candidate neurons. Next, we refine the dictionary using clustering. Finally, we apply the dictionary in order to select neurons and estimate their corresponding activity over time, using a sparse group lasso optimization problem. We apply our proposal to three calcium imaging data sets. 
 
 Our proposed approach is implemented in the \verb=R= package \verb=scalpel=, which is available on \verb=CRAN=.
\end{abstract}



\section{Introduction}
\label{sec:neuronintro}

The field of neuroscience is undergoing a rapid transformation: new technologies are making it possible to image activity in large populations of neurons at cellular resolution in behaving animals \citep{ahrens2013whole, prevedel2014simultaneous, huber2012multiple, dombeck2007imaging}. The resulting \emph{calcium imaging} data sets   promise to provide unprecedented insight into  neural activity. However, they bring with them both statistical and computational challenges.

While calcium imaging data sets have been collected by individual labs for the past several years, up until quite recently large-scale calcium imaging data sets were not publicly-available. Thus, attempts by statisticians to develop  methods for the analysis of these data have been hampered by limited data access.  
 However, in July 2016, the Allen Institute for Brain Science released the Allen Brain Observatory, which contains 30 terabytes of raw data cataloguing 25 mice over 360 different experimental sessions \citep{shen2016brain}.
 This massive data repository is ripe for the development of statistical methods, which can be applied not only to the data from the Allen Institute, but also to calcium imaging data sets collected by individual labs world-wide. 

We now briefly describe the science underlying calcium imaging data. When a neuron fires, voltage-gated calcium channels in the axon terminal open, and calcium floods the cell. Therefore, intracellular calcium concentration  is a surrogate marker for the spiking activity of neurons \citep{grienberger2012imaging}. In recent years, genetically encoded calcium indicators have been developed \citep{rochefort2008calcium,looger2012genetically,chen2013ultrasensitive}. These indicators  bind to intracellular calcium molecules and fluoresce. Thus, the location and timing of neurons firing can be seen through a sequence of two-dimensional images taken over time, typically using two-photon microscopy \citep{svoboda2006principles, helmchen2005deep}.

 A typical calcium imaging video consists of a $500\times 500$ pixels frame over 1 hour, sampled at 15-30 Hz. A given pixel in a given frame is continuous-valued, with larger values representing higher fluorescent intensities due to greater calcium concentrations. An example frame from a calcium imaging video is shown in Figure~\ref{fig:introfig}(a). We have posted snippets of the three calcium imaging videos we analyze at \texttt{www.ajpete.com/software}.

\begin{figure}
\begin{center}
\includegraphics[width=\textwidth]{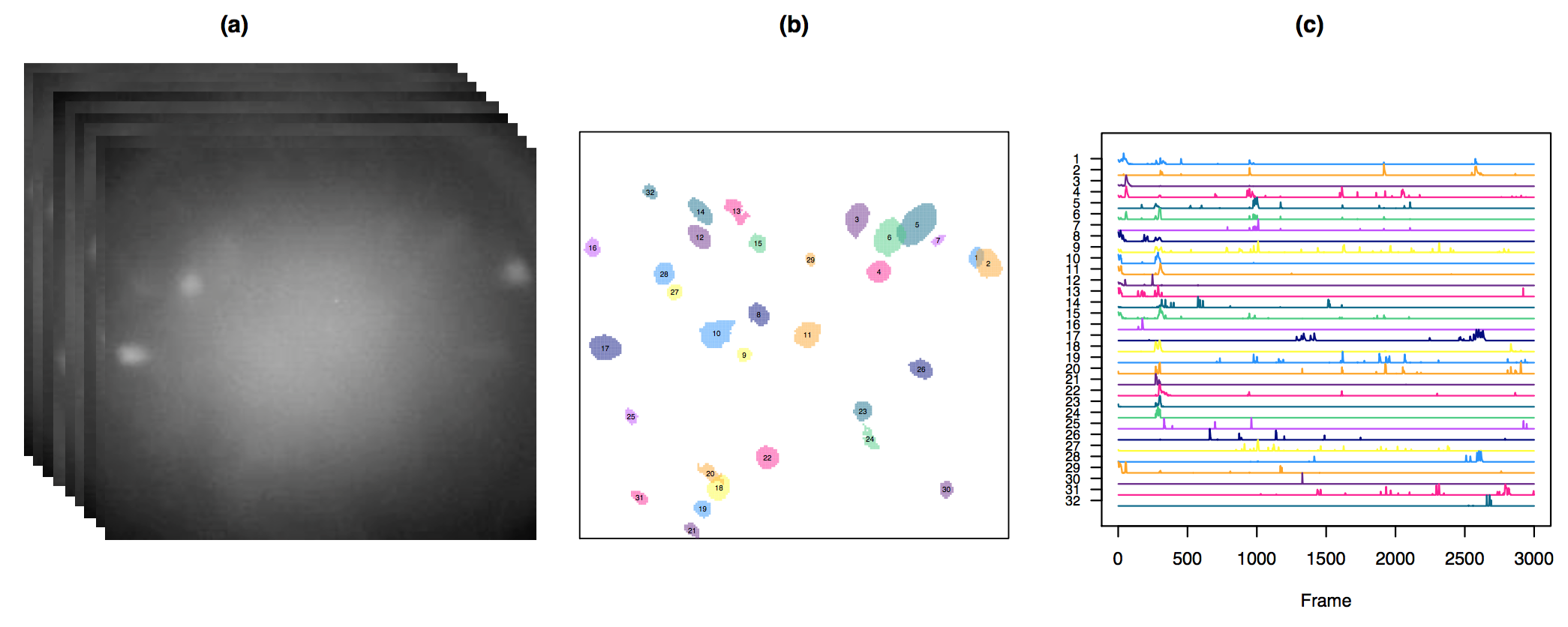}
\end{center}
\caption{In (a), we display sample frames from the raw calcium imaging video described in the text in Section~\ref{sec:neuronmethod}, and analyzed in greater detail in Section~\ref{subsec:onephoton}. We wish to construct a spatial map of the neurons, like that shown in (b), and to estimate the calcium trace for each neuron over time, as shown in (c).}
\label{fig:introfig}
\end{figure}

On the basis of a calcium imaging video, two goals are typically of interest:
\begin{itemize}
\item \emph{Neuron identification:} The goal is to assign pixels of the image frame to neurons. Due to the thickness of the brain slice captured by the imaging technology, neurons can overlap in the two-dimensional image. This means that a single pixel can be assigned to more than one neuron. This step is sometimes referred to as \emph{region of interest identification} or \emph{cell sorting}. An example of neurons identified from a calcium imaging video is shown in Figure~\ref{fig:introfig}(b). 
\item \emph{Calcium quantification:} The goal is to estimate the intracellular calcium concentration for each neuron during each frame of the movie. An example of these estimated \emph{calcium traces} is shown in Figure~\ref{fig:introfig}(c).
\end{itemize}
To a certain extent, these two goals can be accomplished by visual inspection. However, visual inspection suffers from several shortcomings:
\begin{itemize}
\item   It is subjective and it is not reproducible. Two people who view the same video may identify a different set of neurons or different firing times. 
\item It does not yield numerical information regarding neuron firing times, which may be needed for downstream analyses. 
\item It may be inaccurate: for instance, a neuron that is very dim or that fires infrequently may not be identified by visual inspection.
\item It is not feasible on videos with very large neuronal populations or very long durations. In fact, a typical calcium imaging video contains 250,000 pixels and more than 50,000 frames, making visual inspection essentially impossible. 
\end{itemize}

In this paper, we propose a method that  identifies the locations of neurons, and estimates their calcium concentrations over time. Previous proposals to automatically accomplish these tasks have been proposed in the literature, 
and are reviewed in Section~\ref{sec:neuronprevious}. However, our method has several advantages over competing approaches. Unlike many existing approaches, it:
\begin{itemize}
\item Involves few tuning parameters, which are themselves interpretable to the user, can for the most part be set to default values, and can be varied independently;
\item Yields results that are stable across a range of tuning parameters;
\item Is computationally feasible even on very large data sets; and 
\item Uses spatial and temporal information to resolve individual neurons, from sets of overlapping neurons, without post-processing.
\end{itemize}
The methods proposed in this paper can be seen as a necessary step that precedes downstream modeling of calcium imaging data.  For instance, there is substantial interest in modeling functional connectivity among populations of neurons, or using neural activity to decode stimuli (see, e.g., \citet{ko2011functional, mishchencko2011bayesian, paninski2007statistical}). However, before either of those tasks can be carried out, it is necessary to first identify the neurons and determine their calcium concentrations; those are the tasks that we consider in this paper. 

The remainder of this paper is organized as follows. We introduce notation in Section~\ref{sec:neurondata}. In Section~\ref{sec:neuronprevious}, we review related work. We present our proposal in Section~\ref{sec:neuronmethod}, and discuss the selection of tuning parameters  in Section~\ref{sec:neurontuning}. We apply our method to three calcium imaging videos in Section~\ref{sec:neurondataapp}. We discuss the technical details of our proposal  in Section~\ref{sec:neuroncomp}. We close with a discussion in Section~\ref{sec:neurondisc}. Proofs are in the appendix. 


\section{Notation}
\label{sec:neurondata}

Let $P$ denote the total number of pixels per image frame, and $T$ the number of frames of the video. 
We define $\bY$ to be a $P \times T$ matrix for which the $(i,j)$th element, $y_{i,j}$, contains the fluorescence of the $i$th pixel in the $j$th frame. 
 We let $\by_{i,\cdot} = \begin{pmatrix} y_{i,1} & y_{i,2} & \hdots & y_{i,T} \end{pmatrix}^\top$ represent the fluorescence of the $i$th pixel at each of the $T$ frames. We let $\by_{\cdot,j} = \begin{pmatrix} y_{1,j} & y_{2,j} & \hdots & y_{P,j} \end{pmatrix}^\top$ represent the fluorescence of all $P$ pixels during the $j$th frame. We use the same subscript conventions in order  to reference the elements, rows, and columns of other matrices.

The goal of our work is to (1) identify the locations of the neurons; and (2) quantify calcium concentrations for these neurons over time. We view these tasks in the framework of a matrix factorization problem: we decompose $\bY$ into a matrix of spatial components, $\bA \in \mathbb{R}^{P\times K}$, and a matrix of temporal components, $\bZ\in\mathbb{R}^{K\times T}$, such that 
\begin{equation}
\bY\approx\bA\bZ,
\label{eq:decomp}
\end{equation} where $K$ is the total number of estimated neurons.  Note that $\ba_{\cdot, k}$ specifies which of the $P$ pixels of the image frame are mapped to the $k$th neuron, and $\bz_{k, \cdot}$ quantifies the calcium concentration for the $k$th neuron at each of the $T$ video frames. We note that the true number of neurons is unknown, and must be determined as part of the analysis.


\section{Related Work}
\label{sec:neuronprevious}

There are two distinct lines of work in this area. The first focuses on simply identifying the regions of interest in the video, and then subsequently estimating the calcium traces. The second aims to simultaneously identify neurons and quantify their calcium concentrations. 

Methods that focus solely on region of interest identification typically construct a summary image for the calcium imaging video and then segment this image using various approaches. For example, \citet{pachitariu2013extracting} calculate the mean image of the video and then apply convolutional sparse block coding to identify regions of interest. Alternatively, \citet{smith2010parallel} calculate a local cross-correlation image, which is then thresholded using a locally adaptive filter to extract the regions of interest. Similar approaches have been used by others \citep{ozden2008identification, mellen2009semi}. These region of interest approaches often do not fully exploit temporal information, nor do they handle overlapping neurons well due to using summaries aggregated over time.

We now focus on methods proposed to accomplish both goals, neuron identification and calcium quantification, simultaneously. One of the first automatic methods in this area was proposed by \citet{mukamel2009automated}. This method first applies principal component analysis to reduce the dimensionality of the data, followed by spatio-temporal independent component analysis to produce spatial and temporal components that are statistically independent of one another. Though this method is widely used, it often requires heuristic post-processing of the spatial components, and typically fails to distinguish between spatially overlapping neurons \citep{pnevmatikakis2016simultaneous}.

To better handle overlapping neurons, \citet{maruyama2014detecting} proposed a non-negative matrix factorization approach, which estimates $\bA$ and $\bZ$ in \eqref{eq:decomp} by solving
\begin{equation}
\underset{\bA\geq\bzero,\bZ\geq \bzero,\ba_b\geq\bzero}{\text{minimize}}\quad \frac{1}{2}\| \bY - \bA\bZ-\ba_b\bz_b^\top \|_F^2.
\label{eq:neuronmaru}
\end{equation}
In \eqref{eq:neuronmaru}, the term $\ba_b\bz_b^\top$ is a rank-one correction for background noise: 
$\bz_b \in\mathbb{R}^{T}$ is a temporal representation of the background noise (known as the \emph{bleaching line}, and estimated using a linear fit to average fluorescence over time of a background region), and $\ba_b\in\mathbb{R}^{P}$ is a spatial representation of the background noise. Element-wise positivity constraints are imposed on $\bA$, $\bZ$, and $\ba_b$ in \eqref{eq:decomp}. While \eqref{eq:neuronmaru} can handle overlapping neurons, there are no constraints on the sparsity of $\bA$ or $\bZ$, or locality constraints on $\bA$. Thus, the estimated temporal components are very noisy, and the estimated spatial components are often not localized and must be post-processed heuristically. 

To overcome these shortcomings, \citet{haeffele2014structured} modify \eqref{eq:neuronmaru} so that the temporal components $\bA$ are sparse and the spatial components $\bZ$ are sparse and have low total variation. Recently, \citet{pnevmatikakis2016simultaneous} further refine \eqref{eq:neuronmaru} by explicitly modeling the dynamics of the calcium when estimating the temporal components $\bZ$, and  combining a sparsity constraint  with intermediate image filtering when estimating the spatial components $\bA$. \citet{zhou2016efficient} extend the work of \citet{pnevmatikakis2016simultaneous} to better handle one-photon imaging data by (1) modeling the background in a more flexible way and (2) introducing a greedy initialization procedure for the neurons that is more robust to background noise. Related approaches are taken by \citet{diego2013automated, diego2014sparse, friedrich2015fast}.
Other recent approaches consider using convolutional networks trained on manual annotation \citep{apthorpe2016automatic} and multi-level matrix factorization \citep{diego2013learning}.


While these existing approaches show substantial promise, and are a marked improvement over visual identification of neurons from the calcium imaging videos, they also suffer from some shortcomings: 
\begin{itemize}
\item The optimization problems (see, e.g., \eqref{eq:neuronmaru}) are biconvex. Thus, algorithms typically get trapped in unattractive local optima. Furthermore, the results strongly depend on the choice of initialization.
\item Each method involves several  user-selected tuning parameters. There is no natural interpretation to these tuning parameters, which leads to challenges in selection. Furthermore, changing one tuning parameter may necessitate updating all of them. Moreover, there is no natural nesting with respect to the tuning parameters: a slight increase or decrease in one tuning parameter can lead to a completely different set of identified neurons. 
\item The number of neurons $K$ must be specified in advance, and the estimates obtained for different values of $K$ will not be nested: two different values of $K$ can yield completely different answers.
\item Post-processing of the identified neurons is often necessary.
\item Implementation on very large data sets can be computationally burdensome.
\end{itemize}

To overcome these challenges, instead of simultaneously estimating  $\bA$ and $\bZ$ in the model \eqref{eq:decomp}, we take a dictionary learning approach. We first leverage spatial information in order to build a preliminary dictionary of spatial components, which is then refined using a clustering approach to give an estimate of $\bA$. We then use our estimate of $\bA$ in order to obtain an accurate estimate of the temporal components $\bZ$, while simultaneously selecting the final set of neurons in $\bA$. This dictionary learning approach allows us to re-cast \eqref{eq:decomp}, a very challenging unsupervised learning problem, into a much easier supervised learning problem. Compared to existing approaches, our proposal  is much faster to solve computationally, involves more interpretable tuning parameters, and  yields substantially more accurate results. 


\section{Proposed Approach}
\label{sec:neuronmethod}

Our proposed approach is based on dictionary learning. In Figure~\ref{fig:flowchart}, we summarize our  Segmentation, Clustering, and Lasso Penalties (SCALPEL) proposal, which consists of four steps:

\begin{figure}
\begin{center}
\includegraphics[width=\textwidth]{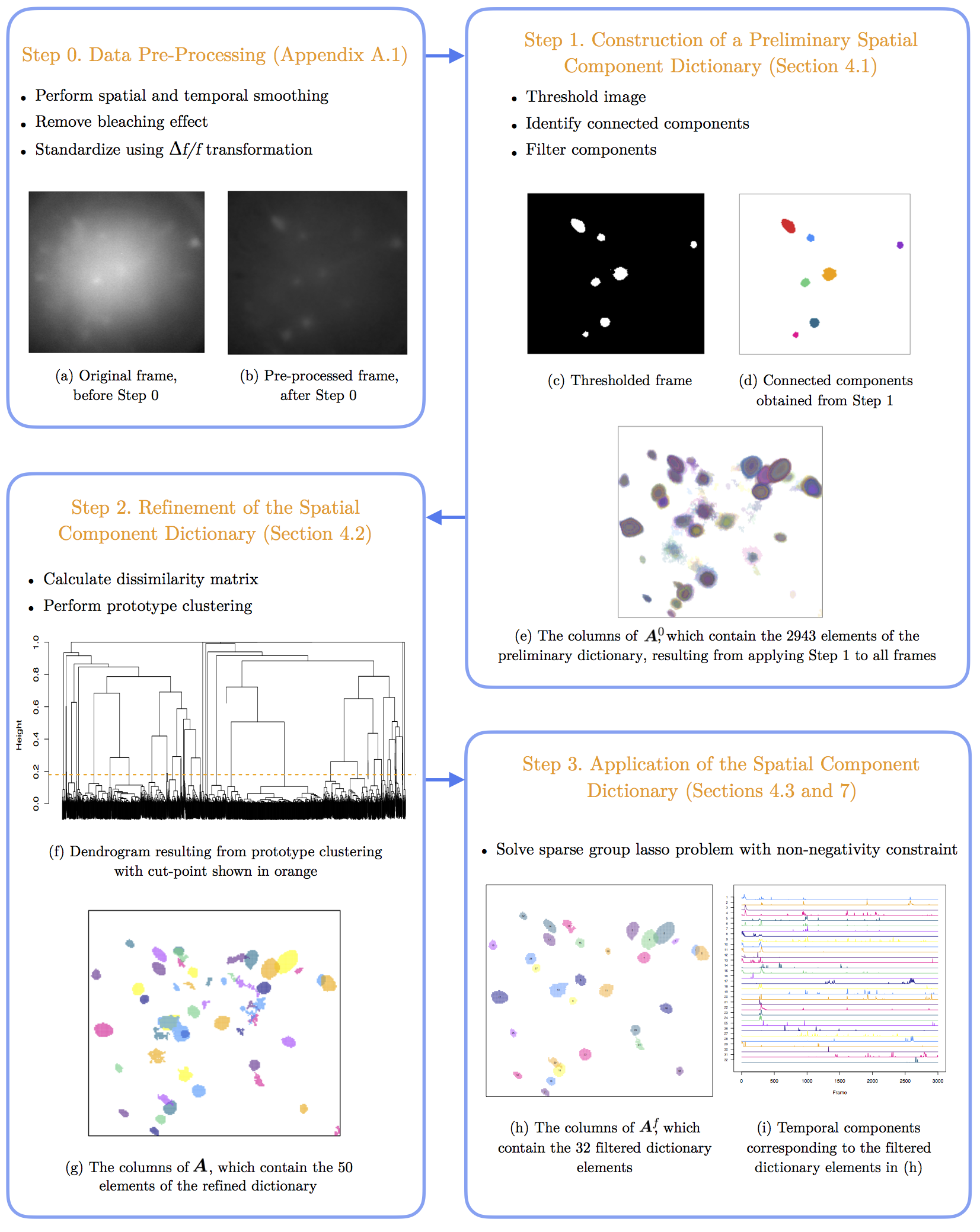}
\end{center}
\caption{A summary of the SCALPEL procedure, along with the results of applying each step to an example data set with $205 \times 226$ pixels and $3000$ frames, described in the text in Section~\ref{sec:neuronmethod}, and analyzed in greater detail in Section~\ref{subsec:onephoton}.}
\label{fig:flowchart}
\end{figure}

\begin{enumerate}
\item[Step 0.] \textit{Data Pre-Processing:} We apply standard pre-processing techniques in order to smooth the data both temporally and spatially, remove the bleaching effect, and calculate a standardized fluorescence. Details are provided in Appendix~\ref{app:neuronpreprocess}. In what follows, $\bY$ refers to the calcium imaging data after these three pre-processing steps have been performed.
\item[Step 1.] \textit{Construction of a Preliminary Spatial Component Dictionary:} We apply a simple image segmentation procedure to each frame of the video in order to derive a spatial component dictionary, which is used to construct the matrix $\bA^0\in\mathbb{R}^{P\times K^0}$ with $a^0_{j,k}=1$  if the $j$th pixel is contained in the $k$th preliminary dictionary element, and $a^0_{j,k}=0$ otherwise. This is discussed further in Section~\ref{subsec:neurondict}.
\item[Step 2.] \textit{Refinement of the Spatial Component Dictionary:} To eliminate redundancy in the preliminary spatial component dictionary, we cluster together dictionary elements that co-localize in time and space. This results in a matrix $\bA \in \mathbb{R}^{P \times K}$: $a_{j,k}=1$  if the $j$th pixel is contained in the $k$th dictionary element, and $a_{j,k}=0$ otherwise. More details are provided in Section~\ref{subsec:neuroncluster}.
\item[Step 3.] \textit{Application of the Spatial Component Dictionary:} We filter out dictionary elements corresponding to clusters with few members, resulting in $\bA^f\in\mathbb{R}^{P\times{K_f}}$, which contains a subset of the columns of $\bA$. We then estimate the temporal components $\bZ$ corresponding to the filtered elements of the dictionary  by solving a sparse group lasso problem with a non-negativity constraint. The $k$th row of $\hat\bZ$ is the estimated calcium trace corresponding to the $k$th filtered dictionary element. For some values of the tuning parameters, entire rows of $\hat\bZ$ may  equal zero. Thus, we are simultaneously selecting the spatial components and estimating the temporal components. Additional details are in Sections~\ref{subsec:neuronsgl} and \ref{sec:neuroncomp}.
\end{enumerate}
Step 1 is applied to each frame separately, and thus can be efficiently performed in parallel across the frames of the video. Similarly, parts of Step 0 can be parallelized across frames and across pixels.

Throughout this section, we illustrate SCALPEL on an example calcium imaging data set that has $205\times226$ pixels and 3000 frames. Figures~\ref{fig:introfig}-\ref{fig:zhou} and Figures~\ref{fig:solnseq}-\ref{fig:maxlam}, as well as Figures~\ref{fig:bleaching}, \ref{fig:preprocess}, \ref{fig:cluster}, and \ref{fig:tradeoff} in the appendix, involve this data set. In Section~\ref{sec:neurondataapp}, we present a more complete analysis of this data set, along with analyses of additional data sets.

\subsection{Step 1: Construction of a Preliminary Spatial Component Dictionary}
\label{subsec:neurondict}

In this step, we identify a large set of preliminary dictionary elements, by applying a simple image segmentation procedure to each frame separately. 

\begin{enumerate}
\item \textit{Threshold Image:} We create a binary image by thresholding the image frame. Figure~\ref{fig:threshold}(b) displays the binary image that results from thresholding the frame shown in Figure~\ref{fig:threshold}(a).
\item \textit{Identify Connected Components:} We identify the connected components of the thresholded image, using the notion of 4-connectivity: connected pixels are pairs of white pixels that are immediately to the left, right, above, or below one another \citep{sonka2014image}. Some of these connected components may represent neurons, whereas others are likely to be noise artifacts or snapshots of multiple nearby neurons. 
\item \textit{Filter Components:} To eliminate noise, we filter components based on their overall size, width, and height. In the examples in this paper, we discard connected components of  fewer than 25 or more than 500 pixels, as well as those with a width or height larger than 30 pixels. 
\end{enumerate}
We now discuss the choice of threshold used  above. After performing Step 0, we expect that the intensities of ``noise pixels" (i.e., pixels that are not part of a firing neuron in that frame) will have a distribution that is approximately symmetric and approximately centered at zero. In contrast, non-noise pixels will have larger values. This implies that the noise pixels should have a value no larger than the negative of the minimum value of $\bY$. Therefore, we threshold each frame using the negative of the minimum value of $\bY$. We also repeat this procedure using a threshold equal to the negative of the 0.1\% quantile of $\bY$, as well as with the average of these two threshold values. In Section~\ref{subsec:neurontuning1}, we discuss alternative approaches to choosing this threshold.

The $K^0$ connected components that arise from performing Step 1 on each frame, for each of the three threshold values, form a \emph{preliminary spatial component dictionary}. We use them to construct the matrix $\bA^0\in\mathbb{R}^{P\times K^0}$: the $k$th column of $\bA^0$ is a vector of $1$'s and $0$'s, indicating whether each pixel is contained in the $k$th preliminary dictionary element.

\begin{figure}
\begin{center}
\includegraphics[width=12cm]{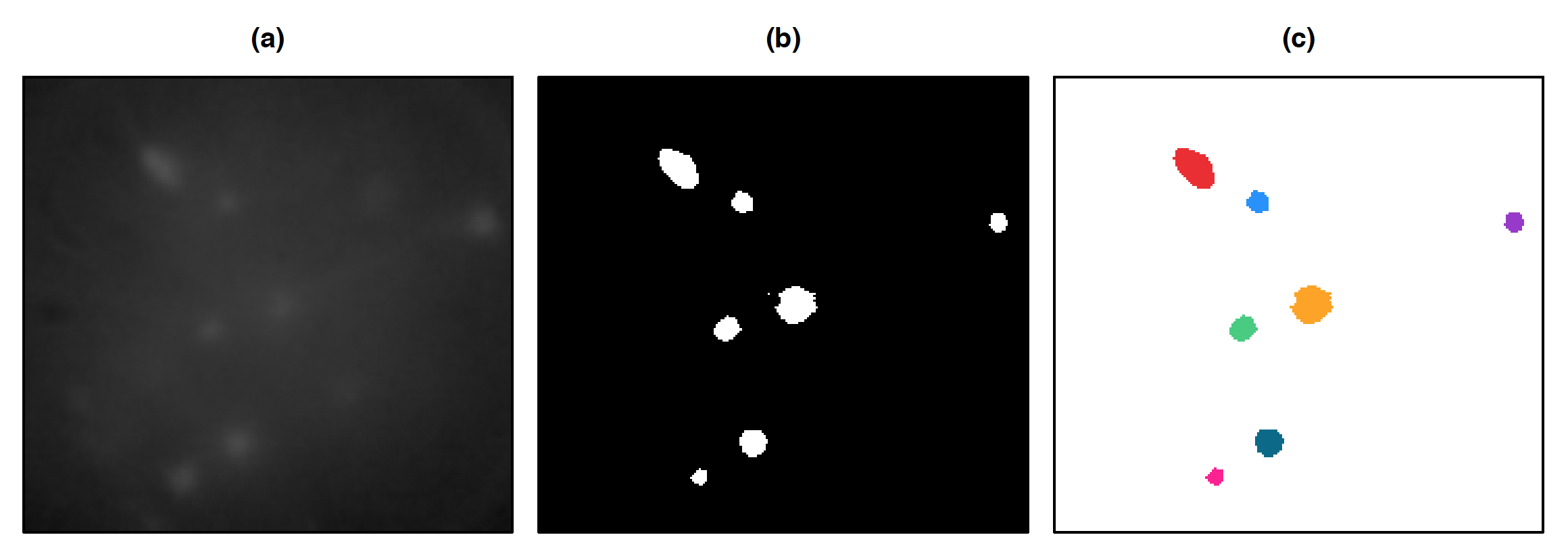}
\end{center}
\caption{In (a), we display a single frame of the example calcium imaging video after performing the pre-processing described in Appendix~\ref{app:neuronpreprocess}. In (b), we show the binary image that results after thresholding using the negative of the 0.1\% quantile of the video's elements. In (c), we display the seven connected components from the image in (b) that contain at least $25$ pixels.}
\label{fig:threshold}
\end{figure}

\subsection{Step 2: Refinement of the Spatial Component Dictionary}
\label{subsec:neuroncluster}
We will now refine the preliminary spatial component dictionary obtained in Step 1, by combining dictionary elements that are very similar to each other, as these likely represent multiple appearances of a single neuron.  We proceed as follows:
\begin{enumerate}
\item \textit{Calculate Dissimilarity Matrix:} We use a novel dissimilarity metric, which incorporates both spatial and temporal information, in order to calculate the dissimilarity between every pair of dictionary elements. More details are given in Section~\ref{subsubsec:neurondissim}.
\item \textit{Perform Prototype Clustering:} We use the aforementioned pair-wise dissimilarities in order to perform  \emph{prototype clustering} of dictionary elements  \citep{bien2011hierarchical}. We also identify a representative dictionary element for each cluster. More details are given in Section~\ref{subsubsec:neuronprotoclust}.
\end{enumerate}
These elements of this refined dictionary make up the columns of the matrix $\bA$, which will be used in Step 3, discussed in Section~\ref{subsec:neuronsgl}.

\subsubsection{Choice of Dissimilarity Metric}
\label{subsubsec:neurondissim}

Before performing clustering, we must decide how to quantify similarity between the $K^0$ elements of the preliminary dictionary obtained in Step 1. Dictionary elements that correspond to the same neuron are likely to have (1) similar spatial maps and (2) similar average fluorescence over time. To this end, we construct a dissimilarity metric that leverages both spatial and temporal information.

 We define $p_{i,j}=(\ba^0_{\cdot, i})^\top\ba^0_{\cdot, j}$, the number of pixels shared between the $i$th and $j$th dictionary elements. When $i=j$, $p_{i,i}$ is simply the number of pixels in the $i$th dictionary element. We then define the spatial dissimilarity between the $i$th and $j$th dictionary elements to be
\begin{equation}
d^s_{i,j}=1-\frac{p_{i, j}}{\sqrt{p_{i,i} p_{j,j}}}.
\label{eq:neuronspatdis}
\end{equation}
Thus, $d^s_{i,j}=1$ if and only if the $i$th and $j$th elements are non-overlapping in space, and $d^s_{i,j}=0$ if and only if they are identical. Note that $d^s_{i,j}$ is known as the cosine dissimilarity or Ochiai coefficient \citep{gower2006similarity}. Alternatives to \eqref{eq:neuronspatdis} are discussed  in Appendix~\ref{app:neuronaltdis}. 
 
We now define the matrix $\bY^B$, a thresholded version of the pre-processed data matrix $\bY$ (obtained in Step 0), with elements of the form
 $$\left[\bY^B\right]_{j,k}=\begin{cases}\left[\bY\right]_{j,k}&\text{if }\left[\bY\right]_{j,k}>-\text{quantile}_{0.1\%}(\bY)\\0&\text{otherwise}\end{cases}.$$
Note that when a value other than the negative of the 0.1\% quantile is used for image segmentation in Step 1, this value can also be used to threshold $\bY$ above. The  temporal dissimilarity between the $i$th and $j$th dictionary elements is defined as 
\begin{equation}
d^t_{i,j}=1-\frac{\left(\ba^0_{\cdot, i}\right)^\top\bY^B \left(\bY^B\right)^\top\ba^0_{\cdot, j}}{\left\|\left(\bY^B\right)^\top\ba^0_{\cdot, i}\right\|_2\left\|\left(\bY^B\right)^\top\ba^0_{\cdot, j}\right\|_2}.\nonumber
\end{equation}
 (Note that the elements of  $\left(\bY^B\right)^\top\ba^0_{\cdot, i}\in\mathbb{R}^T$ represent  the thresholded fluorescence of each time frame, summed over all pixels in the $i$th preliminary dictionary element.) 
We threshold $\bY$ before computing this dissimilarity, because (1) we are interested in the extent to which there is agreement between the peak fluorescences of the $i$th and $j$th preliminary dictionary elements; and  (2) the sparsity induced by thresholding is computationally advantageous.

Finally, the overall dissimilarity is 
\begin{equation}
d_{i,j}=\omega d_{i,j}^s+(1-\omega) d_{i,j}^t,
\label{eq:neuronoveralldis}
\end{equation}
where $\omega\in [0,1]$ controls the relative weightings of the spatial and temporal dissimilarities. We use $\omega=0.2$ to obtain the results shown throughout this paper.  In Figure~\ref{fig:dissimex}, we illustrate pairs of preliminary dictionary elements with various dissimilarities for $\omega=0.2$. 

\begin{figure}
\begin{center}
\includegraphics[width=11cm]{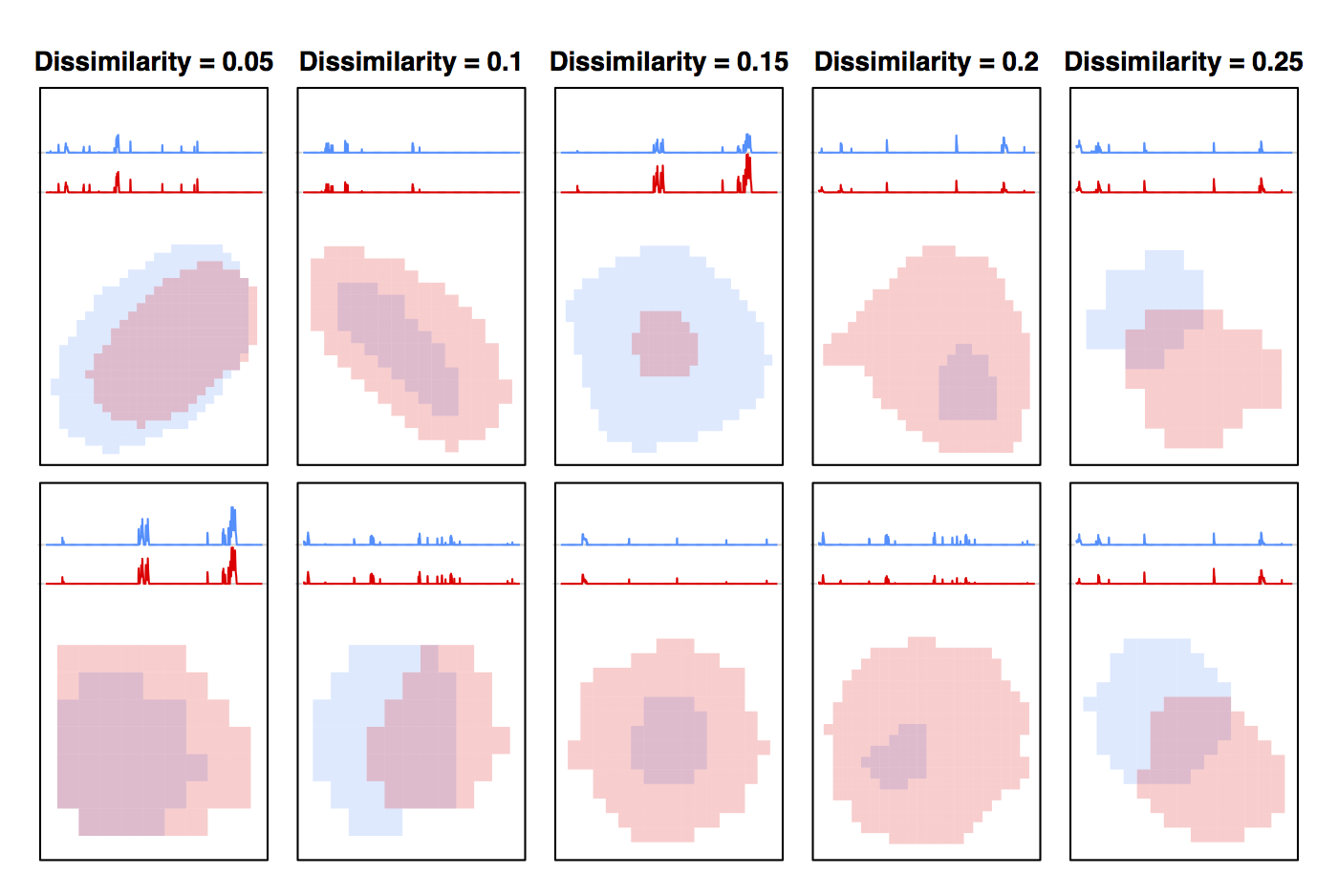}
\end{center}
\caption{Each column displays two pairs of preliminary dictionary elements with overall dissimilarities, as defined in \eqref{eq:neuronoveralldis}, of 0.05, 0.1, 0.15, 0.2, and 0.25. For each preliminary dictionary element,  the average thresholded fluorescence over time and the (zoomed-in) spatial map are shown. These results are based on the example calcium imaging video.}
\label{fig:dissimex}
\end{figure}

\subsubsection{Prototype Clustering}
\label{subsubsec:neuronprotoclust}

We now consider the task of clustering the elements of the preliminary dictionary. To avoid pre-specifying the number of clusters, and to obtain solutions that are nested as the number of clusters is varied, we opt to use hierarchical clustering \citep{trevor2009elements}. 

In particular, we use \emph{prototype clustering}, proposed in \citet{bien2011hierarchical}, with the  dissimilarity given in \eqref{eq:neuronoveralldis}. Prototype clustering guarantees that at least one member of each cluster has a small dissimilarity with all other members of the cluster. To represent each cluster using a single dictionary element, we choose the member with the smallest median dissimilarity to all of the other members. Then we combine the representatives of the $K$ clusters in order to obtain a refined spatial component dictionary. We can represent this refined dictionary with the matrix, $\bA \in \mathbb{R}^{P \times K}$, as follows: $a_{j,k}=1$  if the $j$th pixel is contained in the $k$th cluster's representative, and $a_{j,k}=0$ otherwise.

We apply prototype clustering to the example calcium imaging data set, using the \texttt{R} package \texttt{protoclust} \citep{protoclust}. The resulting dendrogram is  in Figure~\ref{fig:dendro}(a).  In Section~\ref{subsec:neurontuning2}, we discuss choosing the \emph{cut-point}, or  height, at which to cut  the dendrogram. Results for different cut-points are displayed in Figures~\ref{fig:dendro}(b)--(e). An additional example is provided in Appendix~\ref{app:neuroncluster}.

\begin{figure}
\begin{center}
\includegraphics[width=\textwidth]{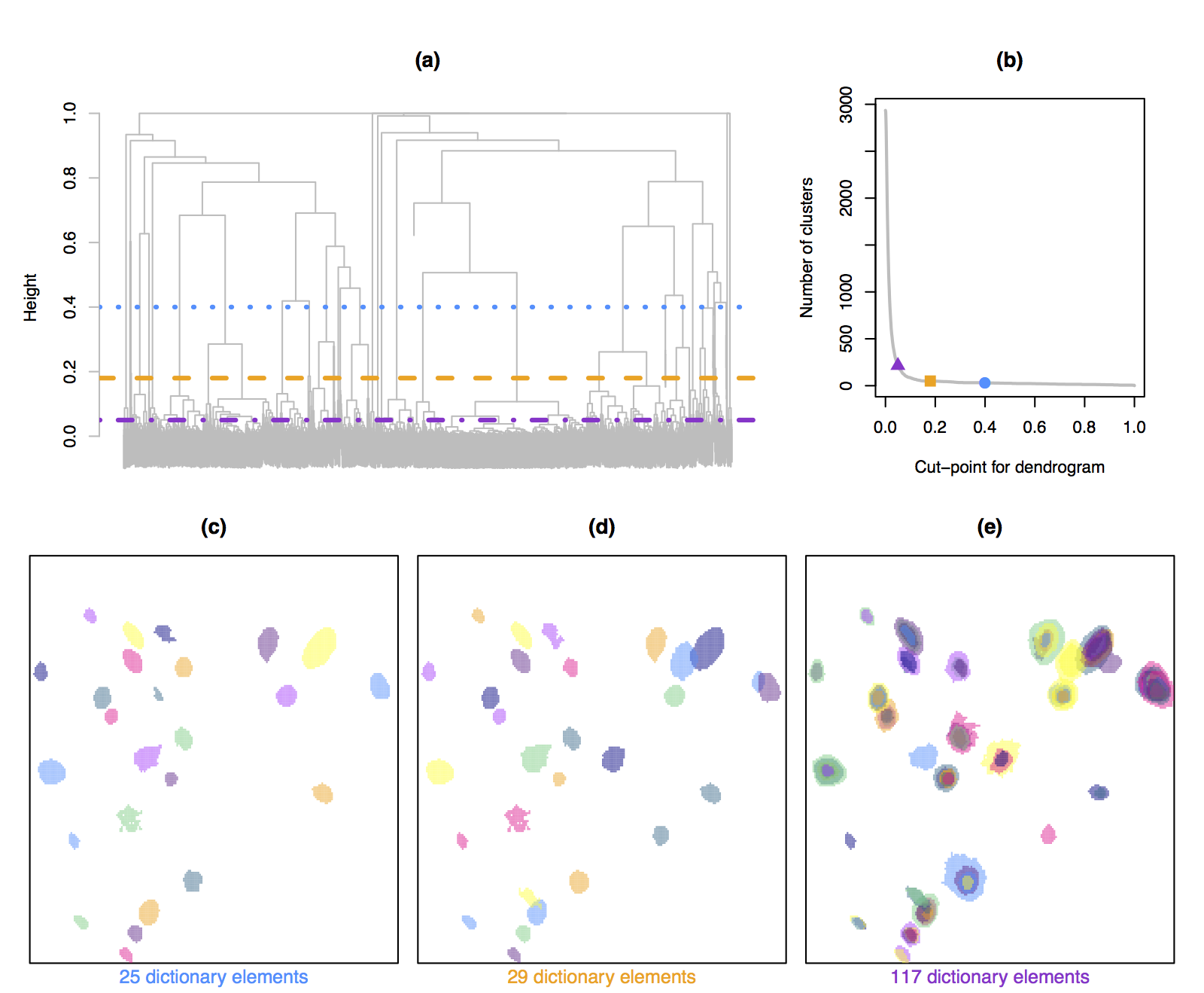}
\end{center}
\caption{In (a), we display the dendrogram that results from applying prototype clustering to the example calcium imaging data set. Three different cut-points are indicated: 0.05 (\protect\includegraphics[height=.18cm]{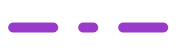}), 0.18 (\protect\includegraphics[height=.18cm]{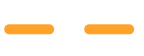}), and 0.4 (\protect\includegraphics[height=.18cm]{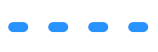}). In (b), we display the number of clusters that result from these three cut-points. In (c)--(e), we show the refined dictionary elements that result from using these three cut-points. For simplicity, we only display dictionary elements corresponding to clusters with at least 5 members.}
\label{fig:dendro}
\end{figure}

\subsection{Step 3: Application of the Spatial Component Dictionary}
\label{subsec:neuronsgl}

In this final step, we optionally filter the $K$ refined dictionary elements, and then estimate the  temporal components associated with this filtered dictionary. There are a couple of ways to perform the optional filtering of the dictionary elements: (1) we can filter based on the number of members in the cluster, as clusters with a larger number of members are more likely to be true neurons, or (2) we can manually filter the elements. This filtering process is discussed more in Appendix~\ref{app:minclusters}. After this filtering, we construct the filtered dictionary, $\bA^f\in\mathbb{R}^{P\times{K_f}}$, which contains the retained columns of $\bA$.  

We estimate the  temporal components associated with the $K_f$ elements  of the final dictionary by solving a sparse group lasso problem with a non-negativity constraint,
\begin{equation}
\label{eq:neuronsgl}
\displaystyle \underset{\bZ\in\mathbb{R}^{{K_f}\times T}, \bZ \geq 0}{\mathrm{minimize}} \quad  \frac{1}{2}  \left \| \bY - \tilde\bA^f\bZ\right\|_F ^2 + \lambda\alpha \sum_{k=1}^{{K_f}} \left\| \bz_{k,\cdot}\right\|_1+\lambda(1-\alpha) \sum_{k=1}^{{K_f}} \left\| \bz_{k,\cdot}\right\|_2,
\end{equation}
where $\alpha\in [0,1]$ and $\lambda>0$ are tuning parameters, and $\tilde\bA^f$ is defined as $\tilde\ba^f_{\cdot, k}=\ba^f_{\cdot, k}/\|\ba^f_{\cdot, k}\|_2^2$. This scaling of $\bA^f$ is justified in Section~\ref{subsec:neuronscaleA}.

The first term of the objective in \eqref{eq:neuronsgl} encourages the spatiotemporal factorization \eqref{eq:decomp} to fit the data closely. The second term in \eqref{eq:neuronsgl} encourages the temporal components to be sparse, so that each neuron is estimated to be active in a small number of frames \citep{tibshirani1996regression}. The third term in \eqref{eq:neuronsgl} encourages group-wise sparsity on the rows of $\bZ$ \citep{yuan2006model}, which leads to selection of dictionary components: for $\lambda (1-\alpha)$ sufficiently large, only a subset of the ${K_f}$ elements in the filtered dictionary will have a non-zero temporal component.

 In Section~\ref{sec:neuroncomp}, we discuss our algorithm for solving \eqref{eq:neuronsgl}, justify the scaling of $\bA^f$ used in \eqref{eq:neuronsgl}, and present conditions under which the solution to \eqref{eq:neuronsgl} is sparse. Readers simply interested in the practical use of our method may disregard this section.


\section{Tuning Parameter Selection}
\label{sec:neurontuning}

SCALPEL involves a number of tuning parameters:
\begin{itemize}
\item \textit{Step 1:} Quantile thresholds for image segmentation
\item \textit{Step 2:} Cut-point for dendrogram 
\item \textit{Step 3:} $\lambda$ and $\alpha$ for Equation~\ref{eq:neuronsgl} 
\end{itemize}
However, in marked contrast to competing methods, the tuning parameters in SCALPEL can be chosen independently at each step and are very interpretable to the user. Furthermore, we strongly recommend the use of default values for all but the choice of $\lambda$ in Step 3.  Therefore, in practice, SCALPEL only involves one tuning parameter that must be selected by the user. We discuss this at greater length below. 

\subsection{Tuning Parameters for Step 1}
\label{subsec:neurontuning1}

The tuning parameters in Step 1 are the quantile thresholds used in image segmentation. In analyzing the example video considered thus far in this paper, we used three different threshold values to segment the video: the negative of the minimum value of $\bY$, the negative of the 0.1\% quantile of $\bY$, and the average of these two values. In principle, these quantile thresholds may need to be adjusted; however, it is straightforward to select reasonable thresholds by visually examining a few frames and their corresponding binary images, as in Figures~\ref{fig:threshold}(a) and (b).

\subsection{Tuning Parameters for Step 2}
\label{subsec:neurontuning2}

In Step 2, we must choose a height $h \in [0,1]$ at which to cut the hierarchical clustering dendrogram, as shown in Figure~\ref{fig:dendro}(a). Fortunately, the cut-point has an intuitive interpretation, which can help guide our choice. If we cut the dendrogram at a height of $h$, then each cluster will contain an element of the preliminary dictionary that has dissimilarity no more than  $h$ with each of the other members of  that cluster. Figure~\ref{fig:dissimex} displays pairs of preliminary dictionary elements with a given dissimilarity. We have used a fixed cut-point of $0.18$ in order to obtain all of the results shown in this paper. An investigator can either choose a cut-point  by visual inspection of the resulting refined dictionary elements, or  can simply use a fixed cut-point, such as $0.18$. We recommend choosing a cut-point less than the dissimilarity weight $\omega$ in \eqref{eq:neuronoveralldis}, as this guarantees that the dictionary elements within each cluster will have spatial overlap. Note that we do not consider $\omega$ to be a tuning parameter, as we recommend keeping it fixed at  $\omega=0.2$.
 
\subsection{Tuning Parameters for Step 3}
\label{subsec:neurontuning3}

In Step 3, we also must choose values of $\lambda$ and $\alpha$ in \eqref{eq:neuronsgl}. We have had empirical success using a fixed value of $\alpha$ near 1, as we expect each neuron to spike a small number of times. We  used  $\alpha=0.9$ to obtain all of the results in this paper. 

Since \eqref{eq:neuronsgl} is a supervised learning problem, cross-validation (or a closely-related validation set approach) is a natural approach to select the value of $\lambda$.  It is computationally feasible, because solving \eqref{eq:neuronsgl} is actually quite efficient, as detailed in Section~\ref{subsec:neuronalg}.  Furthermore, Corollary~\ref{cor:neuronmaxlam} makes it easy to determine the largest value of $\lambda$ that should be considered in our cross-validation scheme. Contrary to standard practice, we use a thresholded version of $\bY$ when calculating the validation error, as we wish to achieve small reconstruction error on the brightest parts of the video, which align with times when neurons are active.
 We illustrate this approach in Section~\ref{sec:neurondataapp}, and provide further details in Appendix~\ref{app:neuroncrossval}. 

As an alternative to cross-validation, the distribution of $\bY$  provides an efficient way to choose $\lambda$. We can think of the elements  of $\bY$ as corresponding to noise pixels (i.e., non-neuronal pixels or neuronal pixels in non-firing frames) as well as  signal pixels. Due to the pre-processing performed in Step 0, the distribution of the noise pixels is approximately symmetric and approximately centered around 0. The distribution of the signal pixels has a much larger center. Therefore, it is likely that many signal pixels, and very few noise pixels, will exceed the negative of (for instance) the 0.1\% quantile of $\bY$. The closed-form solution for non-overlapping neurons given in \eqref{eq:neuronsinglezsoln} implies that  a frame of $\hat\bz$ is zeroed out if its average fluorescence is less than $\lambda\alpha$. Thus, we choose $\lambda=\frac{-\text{quantile}_{0.1\%}(\bY)}{\alpha}$.

Alternatively, we can choose the value of $\lambda$ by visually inspecting the calcium traces (that is, the rows of $\hat\bZ$),  as shown for instance in Figure~\ref{fig:introfig}(c), at each value of $\lambda$. 

Any of these approaches can be used to choose the value of  $\lambda$ in \eqref{eq:neuronsgl}. 
However, in light of the fact that \eqref{eq:neuronsgl} decomposes into groups of non-overlapping neurons as detailed in Section~\ref{subsec:neuronalg}, it is also possible to choose a  different value of $\lambda$ for each group of overlapping neurons. This corresponds to a straightforward modification of the optimization problem \eqref{eq:neuronsgl}.


\section{Results for Calcium Imaging Data}
\label{sec:neurondataapp}
In this section, we compare SCALPEL's performance to that of competitor methods on three calcium imaging data sets. In particular, we consider one-photon and two-photon calcium imaging data sets in Sections~\ref{subsec:onephoton} and \ref{subsec:twophoton}, respectively. In Section~\ref{subsec:neurontiming}, we compare the run times for the various methods. Additional plots regarding the analyses of these three calcium imaging videos are available at \texttt{www.ajpete.com/software}.

\subsection{Application to One-Photon Calcium Imaging Data}
\label{subsec:onephoton}
We now present the results from applying SCALPEL to the calcium imaging video used as an example in Section~\ref{sec:neuronmethod}. This one-photon video, collected by the lab of Ilana Witten at the Princeton Neuroscience Institute, has 3000 frames of size $205\times 226$ pixels sampled at 10 Hz. We used the default tuning parameters to analyze this video. Using the default quantile thresholds that corresponded to thresholds of 0.0544, 0.0743, and 0.0942, Step 1 of SCALPEL resulted in a preliminary dictionary with 2943 elements, which came from 997 different frames of the video. Using a cut-point of 0.18 in Step 2 resulted in a refined dictionary that contained 50 elements. In Step 3, we discarded the twenty-one components corresponding to clusters with fewer than 5 preliminary dictionary elements assigned to them, and then fit the sparse group lasso model with $\alpha=0.9$ and $\lambda=0.0416$, which was chosen using the validation set approach described in Appendix~\ref{app:neuroncrossval}. The results are shown in Figures~\ref{fig:epmAZ}(a) and (b). In Figure~\ref{fig:epmAZ}(c), we compare the estimated neurons to a pixel-wise variance plot of the calcium imaging video. We expect pixels that are part of true neurons to have higher variance than pixels not associated with any neurons. Indeed, we see that many of the estimated neurons coincide with regions of high pixel-wise variance. However, some estimated neurons are in regions with low variance. Examining the frames from which the dictionary elements were derived can provide further evidence as to whether an estimated neuron is truly a neuron. For example, in Figure~\ref{fig:epmframes}(a), we show that one of the estimated neurons in a low-variance region does indeed appear to be a true neuron, while Figure~\ref{fig:epmframes}(b) shows evidence that one of the estimated neurons is not truly a neuron.

Instead of filtering on the basis of cluster size, we can manually classify the elements of the refined dictionary from Step 2 before proceeding to Step 3. In order to decide whether to keep or discard an estimated neuron, we can examine the frames from which the corresponding preliminary dictionary elements were extracted. We performed this process for the 50 dictionary elements from Step 2, and found evidence that 32 of the 50 elements correspond to true neurons. Compared to requiring at least 5 members in a cluster, this manual process produced an additional 4 true neurons, while excluding an element with 11 members that did not appear to correspond to a true neuron (element 22 in Figure~\ref{fig:epmAZ}). The spatial and temporal results based on this manual filtering are shown in Figures~\ref{fig:introfig}(b) and (c).

\begin{figure}
\begin{center}
\includegraphics[width=\textwidth]{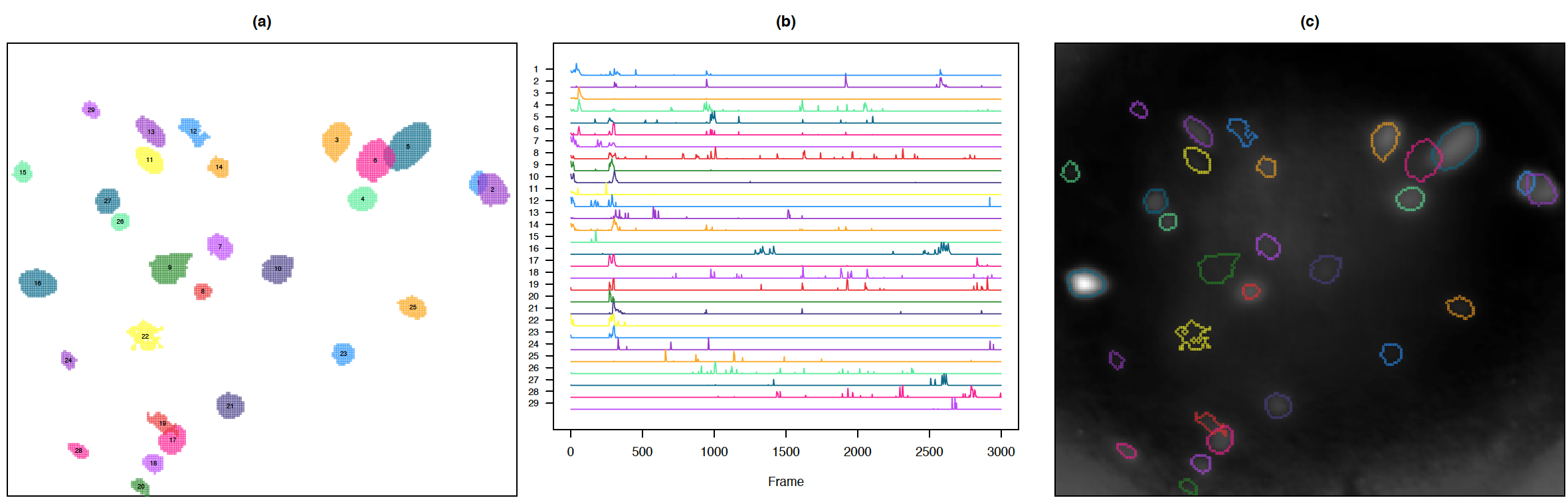}
\end{center}
\caption{In (a), we plot the spatial maps for the 29 elements of the final dictionary for the calcium imaging video considered in Section~\ref{subsec:onephoton}. In (b), we plot their estimated intracellular calcium concentrations corresponding to $\lambda$ chosen via a validation set approach. In (c), we compare the outlines of the 29 dictionary elements from (a) to a heat map of the pixel-wise variance of the calcium imaging video. That is, we plot the variance of each pixel over the 3000 frames, with whiter points indicating higher variance.}
\label{fig:epmAZ}
\end{figure}

\begin{figure}
\begin{center}
\includegraphics[width=11cm]{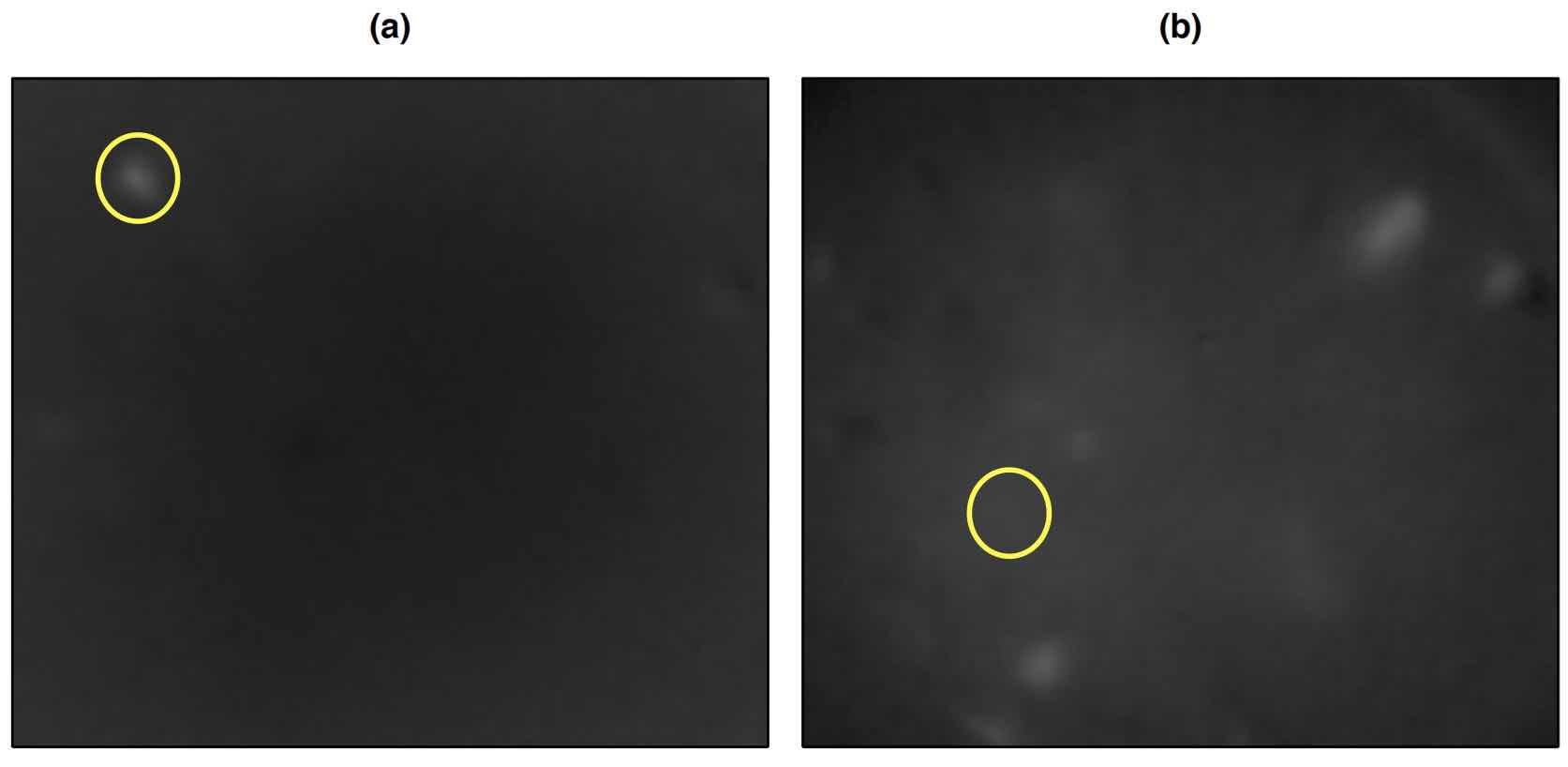}
\end{center}
\caption{In (a), we see that one of the estimated neurons in a low-variance region in Figure~\ref{fig:epmAZ}(c) does correspond to a true neuron. In (b), we see a frame in which one of the estimated neurons was identified, though there does not appear to be a true neuron.}
\label{fig:epmframes}
\end{figure}

We compare the performance of SCALPEL to that of CNMF-E \citep{zhou2016efficient}, which is a proposal for the analysis of one-photon data that takes a matrix factorization approach, as described in Section~\ref{sec:neuronprevious}. The tuning parameters we consider are those noted in Algorithm 1 of  \citet{zhou2016efficient}: the average neuron size $r$ and the width of the 2D Gaussian kernel $\sigma$, which relate to the spatial filtering, and the minimum local correlation $c_{\text{min}}$ and the minimum peak-to-noise ratio $\alpha_{\text{min}}$, which relate to initializing neurons. We choose $r=11$ in accordance with the average diameter of the neurons identified using SCALPEL. The default values suggested for the other tuning parameters are $\sigma=3$, $c_{\text{min}}=0.85$, and $\alpha_{\text{min}}=10$. We present the results for these default values in Figure~\ref{fig:zhou}(a). Only 14 of the 32 neurons identified using SCALPEL were found by CNMF-E using these default parameters. In order to increase the number of neurons found, we consider lower values for $c_{\text{min}}$ and $\alpha_{\text{min}}$. We fit CNMF-E for all combinations of $c_{\text{min}}=0.5,0.6,0.7$ and $\alpha_{\text{min}}=3,5,7$. To assess the performance of these 9 combinations of tuning parameters, we reviewed each estimated neuron for evidence of whether or not it appeared to be a true neuron, based on reviewing the frames in which the neuron was estimated to be most active. We also determined instances where multiple estimated neurons corresponded to the same true neuron.
In Figure~\ref{fig:zhou}(b), we present the fit, chosen from the 9 fits considered, that has the smallest number of false positive neurons (i.e., estimated neurons that are noise or duplicates of other estimated neurons). This fit estimated 24 neurons: 21 elements correspond to neurons identified using SCALPEL, one element (element 23 in Figure~\ref{fig:zhou}(b)) corresponds to a neuron not identified using SCALPEL, and two elements (elements 22 and 24 in Figure~\ref{fig:zhou}(b)) appear to be duplicates of other estimated neurons. 
In Figure~\ref{fig:zhou}(c), we present the fit, chosen from the 9 fits considered, with the highest number of true positive neurons (i.e., neurons that were identified by CNMF-E that appear to be real). This fit estimated 41 neurons: 25 elements correspond to neurons identified using SCALPEL, two elements (elements 27 and 34 in Figure~\ref{fig:zhou}(c)) correspond to neurons not identified using SCALPEL, 11 elements appear to be duplicates of other estimated neurons, and three elements (elements 39, 40, and 41 in Figure~\ref{fig:zhou}(c)) appear to be noise. So while this pair of tuning parameter values resulted in the identification of most of the neurons, it also resulted in a number of false positives. Some of the estimated neurons in Figure~\ref{fig:zhou}(c) are large and diffuse making them difficult to interpret.

\begin{figure}
\begin{center}
\includegraphics[width=\textwidth]{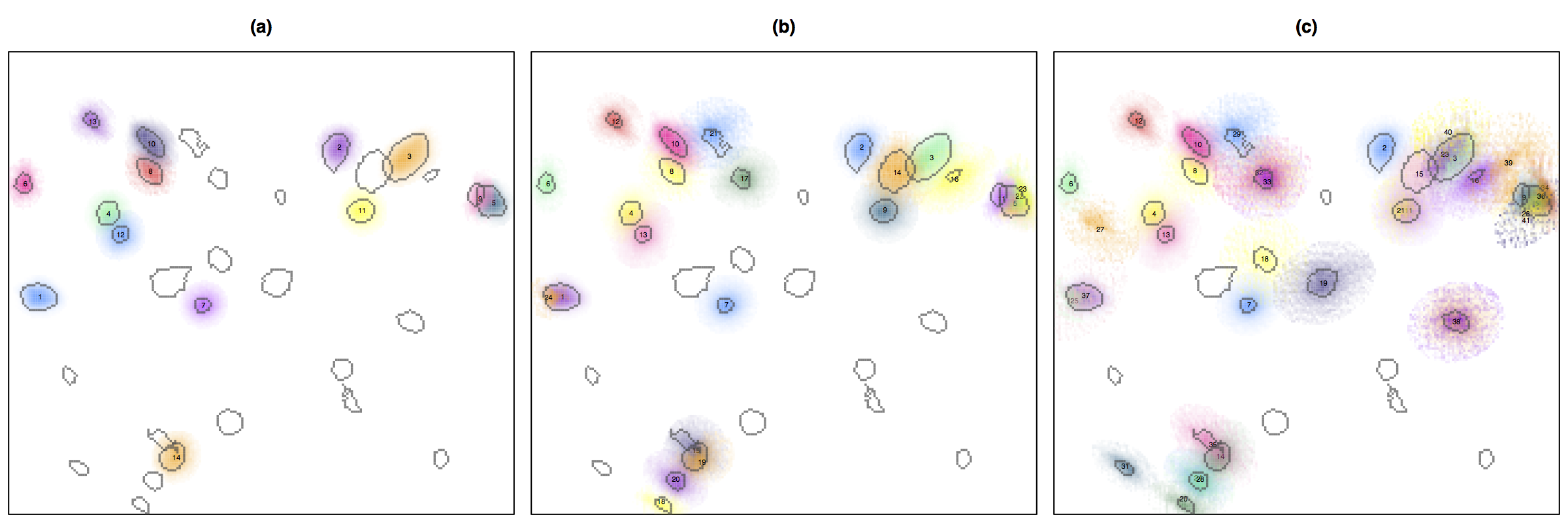}
\end{center}
\caption{We display the estimated neurons that result from applying CNMF-E \citep{zhou2016efficient} to the calcium imaging video considered in Section~\ref{subsec:onephoton} for (a) the default parameters $c_{\text{min}}=0.85$ and $\alpha_{\text{min}}=10$, (b) the parameters $c_{\text{min}}=0.6$ and $\alpha_{\text{min}}=7$, and (c) the parameters $c_{\text{min}}=0.5$ and $\alpha_{\text{min}}=3$. The variation in darkness of the neurons estimated by CNMF-E is due to the fact that they take on continuous values, compared to the binary masks produced by SCALPEL. In each plot, the 32 true neurons identified by SCALPEL are outlined in gray.}
\label{fig:zhou}
\end{figure}

\subsection{Application to Two-Photon Calcium Imaging Data}
\label{subsec:twophoton}

We now illustrate SCALPEL on two calcium imaging videos released by the Allen Institute as part of their Allen Brain Observatory. In addition to releasing the data, the Allen Institute also made available the spatial masks for the neurons they identified in each of the videos. Thus we compare the estimated neurons from SCALPEL to those from the Allen Institute analysis. The two-photon videos we consider are those from experiments 496934409 and 502634578. The videos contain 105,698 and 105,710 frames, respectively, of size $512\times 512$ pixels. In their analyses, the Allen Institute down-sampled the number of frames in each video by 8, which we also did for comparability. For these videos, we found using a threshold value slightly smaller than the negative of the 0.1\% quantile was effective. Otherwise, default values were used for all of the tuning parameters. 

\subsubsection{Allen Brain Observatory Experiment 496934409}
\label{subsubsec:allen1}
Using thresholds of 0.250, 0.423, and 0.596, Step 1 of SCALPEL resulted in a preliminary dictionary with 68,630 elements, which came from 11,739 different frames of the video. After refining the dictionary in Step 2, we were left with 544 elements. In the analysis by the Allen Institute, neurons near the boundary of the field of view were eliminated from consideration. Thus we filtered out 259 elements that contained pixels outside of the region considered by the Allen Institute. Of the remaining 285 elements, 32 of these were determined to be dendrites, 131 were small elements not of primary interest, and 10 were duplicates of other neurons found. Thus in the end, we identified the same 87 neurons that the Allen Institute did, in addition to 25 potential neurons not identified by the Allen Institute. In Figure~\ref{fig:AllenVideo1}(a), we show the neurons jointly identified by SCALPEL and the Allen Institute.  In Figure~\ref{fig:AllenVideo1}(b), we show the potential neurons uniquely identified by SCALPEL, along with evidence supporting them being neurons in Figure~\ref{fig:AllenVideo1}(c).  

\begin{figure}
\begin{center}
\includegraphics[width=\textwidth]{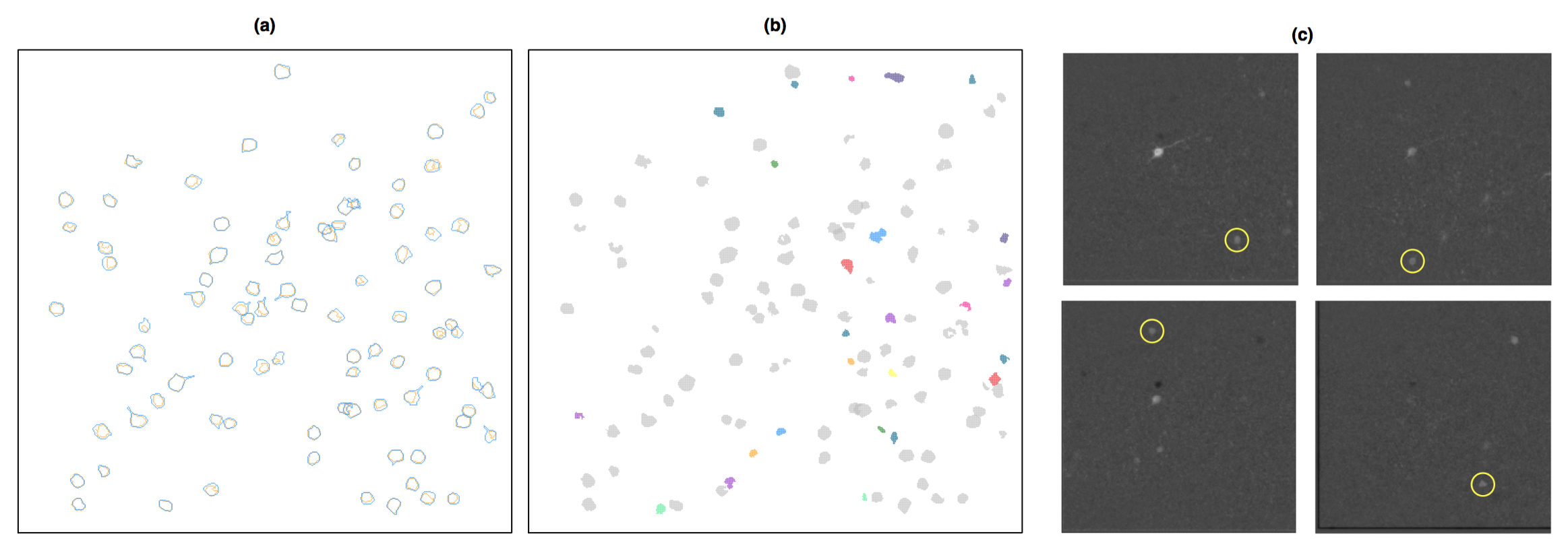}
\end{center}
\caption{We present the results for the calcium imaging video analyzed in Section~\ref{subsubsec:allen1}. In (a), we plot the outlines of the neurons identified by the Allen Institute in blue, along with the outlines of the corresponding SCALPEL neurons in orange. In (b), we plot the 25 potential neurons uniquely identified by SCALPEL in color, along with the SCALPEL neurons also identified by the Allen Institute in gray. In (c), we provide evidence for 4 of the 25 unique neurons. Similar plots for all of the potential neurons uniquely identified by SCALPEL are available at \texttt{www.ajpete.com/software}.}
\label{fig:AllenVideo1}
\end{figure}

\subsubsection{Allen Brain Observatory Experiment 502634578}
\label{subsubsec:allen2}
Using thresholds of 0.250, 0.481, and 0.712, Step 1 of SCALPEL resulted in a preliminary dictionary with 84,996 elements, which came from 12,272 different frames of the video. After refining the dictionary in Step 2, we were left with 1297 elements. Once again, we filtered out the 390 elements that contained pixels outside of the region considered by the Allen Institute. Of the remaining 907 elements, 22 of these were determined to be dendrites, 382 were small elements not of primary interest, and 39 were duplicates of other neurons found. Thus in the end, we identified 370 of the 375 neurons that the Allen Institute did, in addition to 94 potential neurons not identified by the Allen Institute. Note that the 5 neurons identified by the Allen Institute, but not SCALPEL, each appear to be combinations of two neurons. SCALPEL did identify the 10 individual neurons of which these 5 Allen Institute neurons were a combination.  In Figure~\ref{fig:AllenVideo2}(a), we show the neurons jointly identified by SCALPEL and the Allen Institute.  In Figure~\ref{fig:AllenVideo2}(b), we show the potential neurons uniquely identified by SCALPEL, along with evidence supporting them being neurons in Figure~\ref{fig:AllenVideo2}(c).  

\begin{figure}
\begin{center}
\includegraphics[width=\textwidth]{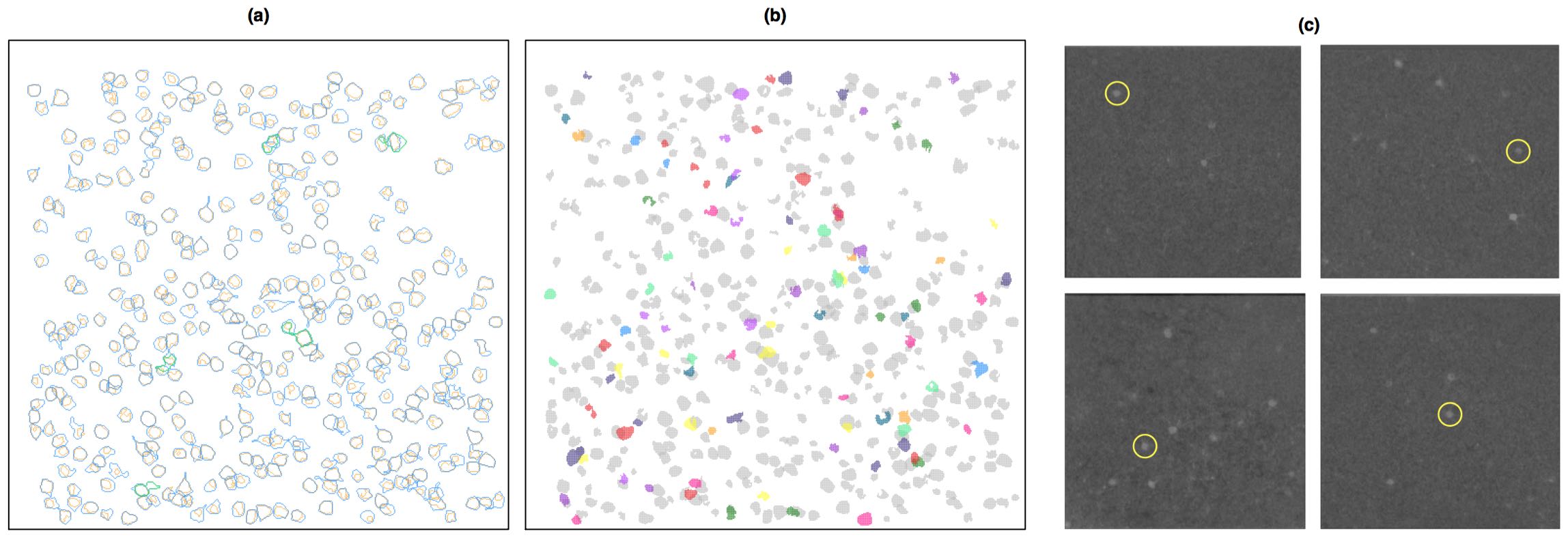}
\end{center}
\caption{We present the results for the calcium imaging video analyzed in Section~\ref{subsubsec:allen2}. In (a), we plot the outlines of the neurons identified by the Allen Institute in blue, along with the outlines of the corresponding SCALPEL neurons in orange. Those shown in green are the Allen Institute neurons that appear to actually be a combination of 2 neurons. In (b), we plot the 94 potential neurons uniquely identified by SCALPEL in color, along with the SCALPEL neurons also identified by the Allen Institute in gray. In (c), we provide evidence for 4 of the 94 unique neurons. Similar plots for all of the potential neurons uniquely identified by SCALPEL are available at \texttt{www.ajpete.com/software}.}
\label{fig:AllenVideo2}
\end{figure}

\subsection{Timing Results}
\label{subsec:neurontiming}

All analyses were run on a Macbook Pro with a 2.0 GHz Intel Sandy Bridge Core i7 processor. Running the SCALPEL pipeline on the one-photon data presented in Section~\ref{subsec:onephoton} took 6 minutes for Step 0 and 2 minutes for Steps 1-3. Running CNMF-E on the one-photon data presented in Section~\ref{subsec:onephoton} took 5, 5, and 7 minutes for the analyses presented in Figures~\ref{fig:zhou}(a), \ref{fig:zhou}(b), and \ref{fig:zhou}(c), respectively. Running the SCALPEL pipeline on the two-photon data presented in Section~\ref{subsubsec:allen1} took 12.85 hours for Step 0, 2.33 hours for Step 1, and 0.42 hours for Step 2. Running the SCALPEL pipeline on the two-photon data presented in Section~\ref{subsubsec:allen2} took 12.50 hours for Step 0, 2.55 hours for Step 1, and 0.43 hours for Step 2. 

Further computational gains could be made by parallelizing the implementation of SCALPEL Steps 0 and 1. In the instance of the analysis of the Allen Institute data, this could provide an almost 10-fold decrease in the computational time required. Also, recall that SCALPEL's most time-intensive step, Step 0, is only ever run a single time for each data set, regardless of whether the user wishes to fit SCALPEL for different tuning parameters.


\section{Further Discussion of Step 3}
\label{sec:neuroncomp}

In this section, we elaborate on issues related to solving the sparse group lasso problem \eqref{eq:neuronsgl} in Step 3. The discussion in this section is somewhat technical, and can be skipped by readers only interested in the practical use of SCALPEL. We discuss the solution to \eqref{eq:neuronsgl} when ${K_f}=1$ in Section~\ref{subsec:neuronsinglecomp}, our algorithm for solving \eqref{eq:neuronsgl} for any value of ${K_f}$ in Section~\ref{subsec:neuronalg}, the justification for the scaling of $\bA^f$ in Section~\ref{subsec:neuronscaleA}, a result about the tuning parameters $\alpha$ and $\lambda$ that lead to a sparse solution in Section~\ref{subsec:neuronsparse}, and the ability of the group lasso penalty in  \eqref{eq:neuronsgl} to zero out unwanted dictionary elements in Section~\ref{subsec:neurondouble}.

\subsection{Single Component Problem}
\label{subsec:neuronsinglecomp}

We first consider solving \eqref{eq:neuronsgl} in the setting in which there is a single spatial component (${K_f}=1$). While calcium imaging data will not have only a single neuron, this setting provides intuition, and will prove useful when we later solve \eqref{eq:neuronsgl} for ${K_f}>1$ in Section~\ref{subsec:neuronalg}.

\begin{lemma}
\label{lem:neurononecomp}
The solution to
\begin{equation}
\label{eq:neurononecomp}
\displaystyle \underset{\bz\in\mathbb{R}^{T},\bz\geq \bzero}{\mathrm{minimize}} \quad  \frac{1}{2}  \left \| \bY - \tilde\ba^f\bz^\top\right\|_F ^2 + \lambda\alpha\left\| \bz\right\|_1 +\lambda(1-\alpha)  \left\| \bz\right\|_2 
\end{equation}
is
\begin{equation}
\label{eq:neuronsinglezsoln}
\hat\bz=\left( 1-\frac{\lambda(1-\alpha)}{\left\|\left( \bY^\top\tilde\ba^f-\lambda\alpha\bone \right)_+\right\|_2}\right)_+ \left( \frac{\bY^\top\tilde\ba^f -\lambda\alpha\bone}{(\tilde\ba^f)^\top\tilde\ba^f}\right)_+,
\end{equation}
where the positive part operator is applied element-wise.
\end{lemma}
The proof of Lemma~\ref{lem:neurononecomp} is in Appendix~\ref{app:neurononecomp}. We can inspect the solution \eqref{eq:neuronsinglezsoln} to gain intuition. Recall that $\tilde\ba^f_{\cdot,k}\equiv\ba^f_{\cdot,k}/\| \ba^f_{\cdot,k}\|_2^2$, where $\ba^f_{\cdot,k}$ has binary elements. Therefore, $\bY^\top\tilde\ba^f\in\mathbb{R}^{T}$ is the average fluorescence of pixels in the filtered dictionary element at each of the frames and $\frac{1}{(\tilde\ba^f)^\top\tilde\ba^f}$ equals the number of pixels in the dictionary element. When $\lambda=0$, $\hat\bz=\left(\frac{\bY^\top\tilde\ba^f}{(\tilde\ba^f)^\top\tilde\ba^f}\right)_+$, which is simply the positive part of the total fluorescence at all pixels in the dictionary element over time. We now consider the impact of $\lambda$ for three different settings of $\alpha$:
\begin{itemize}
\item \textit{$\alpha=1$:} In this setting, $\hat\bz$ is the positive part of the soft-thresholded total fluorescence. This encourages elements of $\hat\bz$ to be exactly zero for frames in which the dictionary element has low fluorescence. 
\item \textit{$\alpha=0$:} In this setting, $\hat\bz$ is found by scaling all elements of the total fluorescence by the same amount. Thus, individual elements of $\hat\bz$ are not encouraged to be 0, though $\hat\bz=\bzero$ if the dictionary element has a low amount of fluorescence across all frames (i.e., $\|\bY^\top\tilde\ba^f\|_2$ small) or $\lambda$ is very large.
\item \textit{$\alpha\in(0,1)$:} Both soft-thresholding and soft-scaling are performed, which encourages sparsity of individual elements of $\hat\bz$ and the entire vector $\hat\bz$, respectively.
\end{itemize}
In Figure~\ref{fig:solnseq}, we illustrate the values of $\hat\bz$ for the three scenarios described above.

\begin{figure}
\begin{center}
\includegraphics[width=\textwidth]{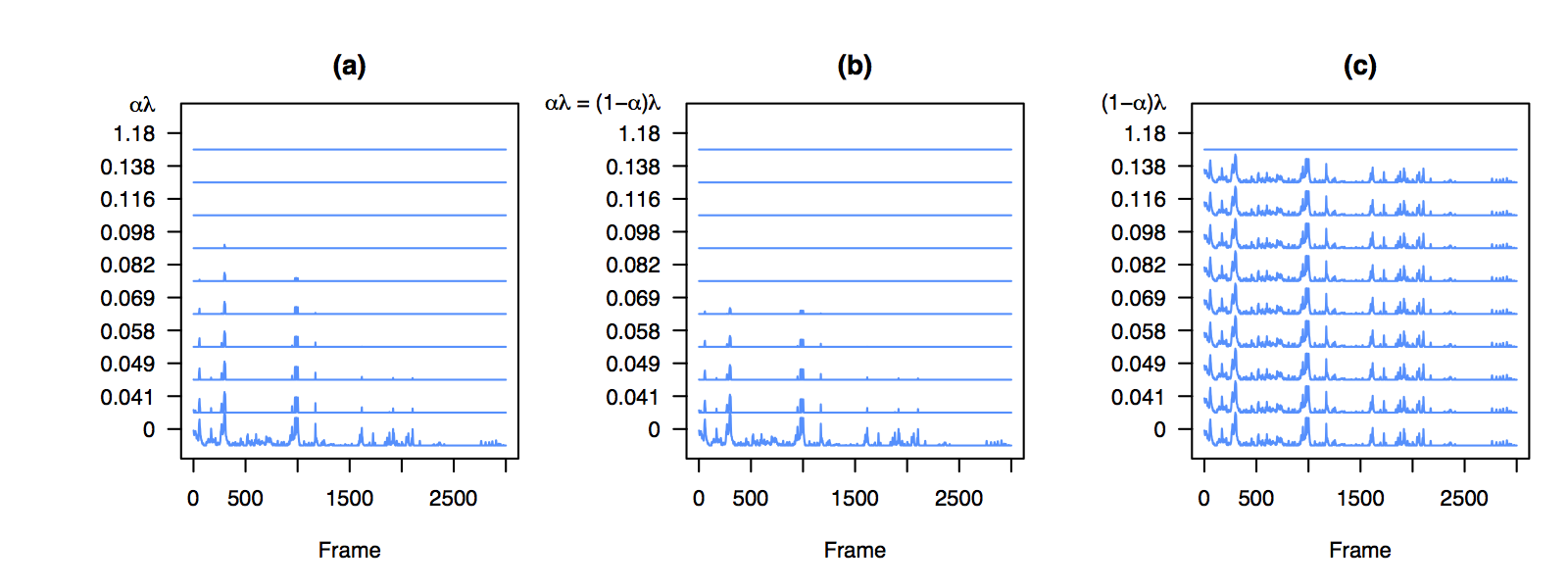}
\end{center}
\caption{For a single dictionary element in the example video, we plot the solution $\hat\bz$, as given in \eqref{eq:neuronsinglezsoln}, for a range of $\lambda$ when (a) only a lasso penalty is used ($\alpha=1$), (b) a mixture of penalties is used ($\alpha=0.5$), and (c) only a group lasso penalty is used ($\alpha=0$).}
\label{fig:solnseq}
\end{figure}

\subsection{Algorithm}
\label{subsec:neuronalg}

We now consider how to solve \eqref{eq:neuronsgl} for ${K_f}>1$. While generalized gradient descent \citep{beck2009fast} can be used to solve concurrently for $\bz_{1,\cdot},\ldots,\bz_{{K_f},\cdot}$ in \eqref{eq:neuronsgl}, the problem is solved more efficiently by noting that \eqref{eq:neuronsgl} is decomposable into groups of overlapping spatial components.

Let $\mathcal{N}_1,\ldots,\mathcal{N}_S$ denote a partition of the ${K_f}$ elements of the filtered dictionary, such that $\mathcal{N}_s \cap \mathcal{N}_{s'} = \emptyset$ for $s \neq s'$, and $\cup_{s=1}^S \mathcal{N}_s = \{1,\ldots,{K_f}\}$. 
Define the mapping $\mathcal{M}(\mathcal{N}_s) = \{ p \in (1,\ldots,P): (\tilde{\bA^f}_{p,\mathcal{N}_s})^\top \bone > 0\}$.  That is, $\mathcal{M}(\mathcal{N}_s)$ indexes the set of pixels that are active in that subset of neurons. 
\begin{lemma}
Suppose that $\mathcal{M}(\mathcal{N}_s) \cap \mathcal{M}(\mathcal{N}_{s'}) = \emptyset$ for all $s \neq s'$, so that there is no spatial overlap between the sets of filtered dictionary elements $\mathcal{N}_1,\ldots,\mathcal{N}_{S}$. 
Then solving \eqref{eq:neuronsgl} gives the same solution as solving
\begin{equation}
\label{eq:neuronreducesgl}
\displaystyle \underset{\bZ_{\mathcal{N}_s, \cdot}\in\mathbb{R}_+^{\vert \mathcal{N}_s\vert \times T}}{\mathrm{minimize}} \quad  \frac{1}{2}  \left\| \bY_{\mathcal{M}(\mathcal{N}_s), \cdot} - \tilde\bA^f_{\mathcal{M}(\mathcal{N}_s),\mathcal{N}_{s}}\bZ_{\mathcal{N}_{s}, \cdot}\right\|_F^2+ \lambda\alpha \sum_{n\in\mathcal{N}_s} \left\| \bz_{n,\cdot}\right\|_1+\lambda(1-\alpha) \sum_{n\in\mathcal{N}_s} \left\| \bz_{n,\cdot}\right\|_2,
\end{equation}
for $s=1,2,\ldots,S$.
\label{lem:neurondecomp}
\end{lemma}

The proof of Lemma~\ref{lem:neurondecomp} is in Appendix~\ref{app:neurondecomp}.

Our approach to solving \eqref{eq:neuronreducesgl} depends on the size of $\mathcal{N}_s$. For $\vert \mathcal{N}_s\vert=1$, we can simply use the closed-form solution for $\bz_{\mathcal{N}_s, \cdot}$ given by Lemma~\ref{lem:neurononecomp}. This is advantageous as the calcium imaging data sets that we have analyzed often have some dictionary elements that do not overlap with any others. For $\vert\mathcal{N}_s\vert>1$, we use generalized gradient descent to solve for the global optimum of \eqref{eq:neuronreducesgl} \citep{beck2009fast}. 

In light of Lemma~\ref{lem:neurondecomp}, in order to solve \eqref{eq:neuronsgl}, we first partition the filtered dictionary elements into $S$ sets, $\mathcal{N}_1,\ldots,\mathcal{N}_S$, such that there is no overlap between the pixels in the $S$ sets, and so that no set can be partitioned further. This can be done quickly, as outlined in Step 1 of Algorithm~\ref{alg:neuronfastalg}. Then, we solve \eqref{eq:neuronreducesgl} for $s=1,\ldots,S$. Details are provided in Algorithm~\ref{alg:neuronfastalg}. 

We typically solve \eqref{eq:neuronsgl} for a sequence of exponentially decreasing $\lambda$ values. To improve computational performance, Step 2(b) of Algorithm~\ref{alg:neuronfastalg} can be implemented using warm starts, in which $\bZ_{\mathcal{N}_s, \cdot}^{(0)}$ is initialized as the solution for $\bZ_{\mathcal{N}_s, \cdot}$ for the previous value of $\lambda$. Additional details regarding the derivation of the generalized gradient descent algorithm used in Step 2(b) of Algorithm~\ref{alg:neuronfastalg} are given in Appendix~\ref{app:neuronappggd}.

\begin{algorithm}
\caption{--- Algorithm for Solving Equation \eqref{eq:neuronsgl}} \vskip .2 cm
\label{alg:neuronfastalg}
\begin{enumerate}
\item Construct the adjacency matrix $\bN\in\mathbb{R}^{{K_f}\times {K_f}}$ with $n_{i,j}=\begin{cases}1 &\text{if }(\ba^f_{\cdot,i})^\top\ba^f_{\cdot,j}>0\\0&\text{if }(\ba^f_{\cdot,i})^\top\ba^f_{\cdot,j}=0 \end{cases}.$ Let $\mathcal{N}_1,\mathcal{N}_2,\ldots,\mathcal{N}_S$ denote the connected components of the graph corresponding to $\bN$. That is, $\mathcal{N}_s$ indexes the filtered dictionary elements in the $s$th connected component. Define the mapping $\mathcal{M}(\mathcal{N}_s)=\{ p\in(1,\ldots,P) : (\tilde\bA^f_{p, \mathcal{N}_s})^\top\bone>0\}$. \vskip .2 cm
\item For $s=1,2,\ldots,S$, solve
\begin{equation}
\label{eq:neuronsgllittle}
\displaystyle \underset{\bZ_{\mathcal{N}_s, \cdot}\geq 0}{\mathrm{minimize}} \quad  \frac{1}{2}  \left\| \bY_{\mathcal{M}(\mathcal{N}_s), \cdot} - \tilde\bA^f_{\mathcal{M}(\mathcal{N}_s),\mathcal{N}_{s}}\bZ_{\mathcal{N}_{s}, \cdot}\right\|_F^2+ \lambda\alpha \sum_{n\in\mathcal{N}_s} \left\| \bz_{n,\cdot}\right\|_1+\lambda(1-\alpha) \sum_{n\in\mathcal{N}_s} \left\| \bz_{n,\cdot}\right\|_2,
\end{equation}
using one of the two following approaches:
\begin{enumerate}
\item By Lemma~\ref{lem:neurononecomp}, if $\vert \mathcal{N}_s\vert=1$, the closed-form solution for $\bz_{\mathcal{N}_s ,\cdot}$ in \eqref{eq:neuronsgllittle} is
\begin{equation}
\label{eq:neuronsinglezsolnalg}
\hat\bz_{\mathcal{N}_s, \cdot}=\left( 1-\frac{\lambda(1-\alpha)}{\left\|\left( (\bY_{\mathcal{M}(\mathcal{N}_s),\cdot})^\top\tilde\ba^f_{\mathcal{M}(\mathcal{N}_s),\mathcal{N}_s} -\lambda\alpha\bone \right)_+\right\|_2}\right)_+ \left( \frac{(\bY_{\mathcal{M}(\mathcal{N}_s),\cdot})^\top\tilde\ba^f_{\mathcal{M}(\mathcal{N}_s),\mathcal{N}_s} -\lambda\alpha\bone}{(\tilde\ba^f_{\mathcal{M}(\mathcal{N}_s),\mathcal{N}_s})^\top\tilde\ba^f_{\mathcal{M}(\mathcal{N}_s),\mathcal{N}_s}}\right)_+.
\end{equation}
\item If $\vert \mathcal{N}_s\vert>1$, use generalized gradient descent to solve \eqref{eq:neuronsgllittle} for $\bZ_{\mathcal{N}_s, \cdot}$:
\begin{enumerate}
\item Let $f\left(\bZ_{\mathcal{N}_s, \cdot}\right)= \frac{1}{2}  \left \| \bY_{\mathcal{M}(\mathcal{N}_s),\cdot} - \tilde\bA^f_{\mathcal{M}(\mathcal{N}_s),\mathcal{N}_s}\bZ_{\mathcal{N}_s ,\cdot}\right\|_F ^2 + \lambda\alpha \bone^\top\bZ_{\mathcal{N}_s, \cdot}\bone$. Initialize $\bZ_{\mathcal{N}_s, \cdot}^{(0)} := \bzero$ and let 
$t:=(\max_{n\in\mathcal{N}_s} \sum_{j\in\mathcal{N}_s} (\tilde\ba^f_{\mathcal{M}(\mathcal{N}_s),j})^\top\tilde\ba^f_{\mathcal{M}(\mathcal{N}_s),n})^{-1}$. \vskip .2 cm
\item For $b=1,2,\ldots$, until convergence, iterate:\vskip .25 cm
\hskip 1cm $\nabla f\left(\bZ_{\mathcal{N}_s, \cdot}^{(b-1)}\right):=-(\tilde\bA^f_{\mathcal{M}(\mathcal{N}_s),\mathcal{N}_s})^\top\left(\bY_{\mathcal{M}(\mathcal{N}_s),\cdot} - \tilde\bA^f_{\mathcal{M}(\mathcal{N}_s),\mathcal{N}_s}\bZ_{\mathcal{N}_s, \cdot}^{(b-1)} \right) +\lambda\alpha\bone\bone^\top$,\\
\hskip 1cm $\tilde\bY_{\mathcal{N}_s,\cdot}:=\bZ_{\mathcal{N}_s, \cdot}^{(b-1)}-t\nabla f\left(\bZ_{\mathcal{N}_s ,\cdot}^{(b-1)}\right)$, and\\ 
\hskip 1cm $\bz^{(b)}_{n,\cdot}:= \left(1-\frac{\lambda(1-\alpha)t}{\left\| \left( \tilde\by_{n,\cdot}\right)_+\right\|_2} \right)_+ \left( \tilde\by_{n,\cdot}\right)_+$ for $n\in\mathcal{N}_s$.
\end{enumerate}
\end{enumerate}
\end{enumerate}
\end{algorithm}

\subsection{Scaling of $\bA^f$}
\label{subsec:neuronscaleA}

In \eqref{eq:neuronsgl}, the $k$th column of the matrix $\tilde\bA^f$ encodes the spatial mapping of the $k$th filtered dictionary element, after scaling. 
To obtain $\tilde\bA^f$, we divide the $k$th column of $\bA^f$ by $\| \ba^f_{\cdot, k}\|_2^2$, the number of pixels in the $k$th filtered dictionary element. This scaling is performed so that the sizes of the dictionary elements do not impact when the components enter the model. That is, we would like $\| \ba^f_{\cdot, k}\|_2^2$ to be independent of the largest value of $\lambda$ for which $\hat\bz_{k,\cdot}\neq\bzero$. The following lemma supports this particular scaling of the columns of $\bA^f$. 

\begin{lemma}
\label{lem:Ascale}
Suppose $\bY=\bA^f\bZ^*$ where the following conditions hold:
\begin{enumerate}
\item[(i)] $\bA^f\in\mathbb{R}^{P\times {K_f}}$ with $(\ba^f_{\cdot, 1})^\top\ba^f_{\cdot, 2}=0$, $(\ba^f_{\cdot, 1})^\top\ba^f_{\cdot, k}=0$ for $k=3,\ldots,{K_f}$, and $(\ba^f_{\cdot, 2})^\top\ba^f_{\cdot, k}=0$ for $k=3,\ldots,{K_f}$ and
\item[(ii)] $\bZ^*\in\mathbb{R}^{{K_f}\times T}$ with $\bz^*_{1, \cdot}=\bP\bz^*_{2, \cdot}$ for some $T\times T$ permutation matrix $\bP$.
\end{enumerate}
If we solve \eqref{eq:neuronsgl} for $\bZ$ with $\tilde\bA^f$ such that $\tilde\ba^f_{\cdot, k}=\ba^f_{\cdot, k}/\| \ba^f_{\cdot, k}\|_2^2$, then $\hat\bz_{1, \cdot}=\bzero$ if and only if $\hat\bz_{2, \cdot}=\bzero$.
\end{lemma}

The proof of Lemma~\ref{lem:Ascale} is in Appendix~\ref{app:neuronAscale}. Lemma~\ref{lem:Ascale} indicates that two non-overlapping spatial components, possibly of different sizes, whose temporal components are identical up to a permutation, will enter the model at the same value of $\lambda$. In Figure~\ref{fig:scaleA}, we provide empirical evidence for the chosen scaling of $\bA^f$. 

\begin{figure}
\begin{center}
\includegraphics[width=\textwidth]{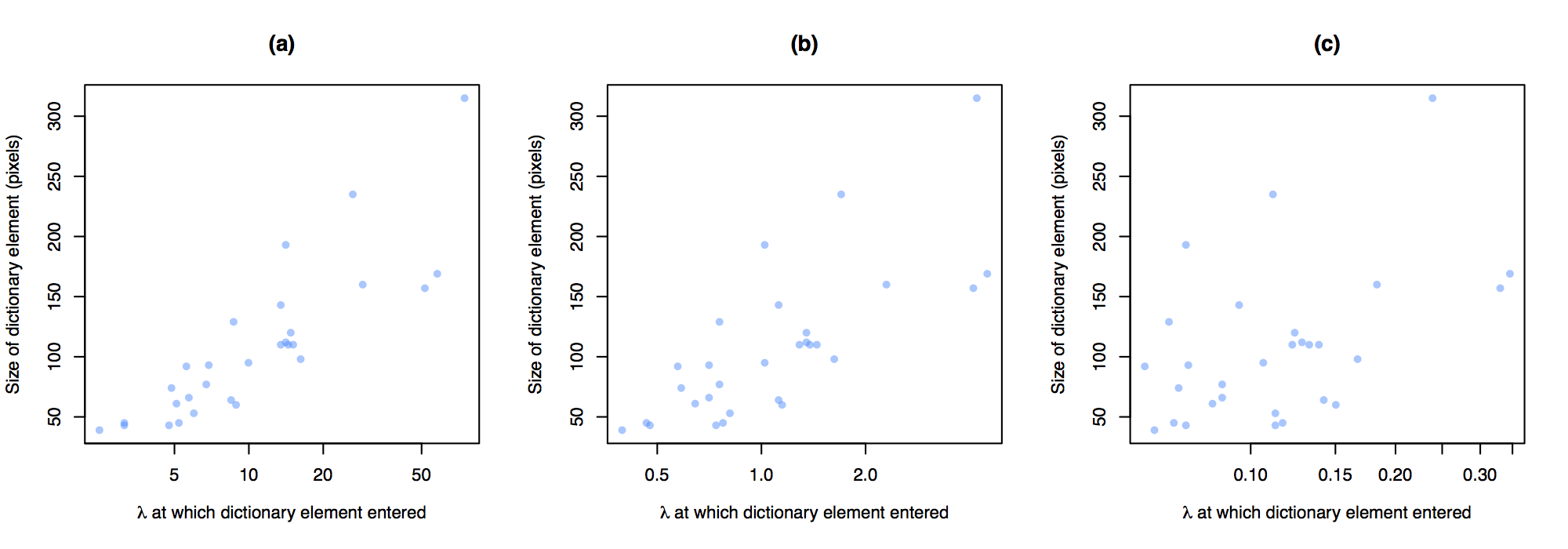}
\end{center}
\caption{We solve \eqref{eq:neuronsgl} for different scalings of $\tilde\bA^f$. For each spatial component, we note the value of $\lambda$ at which the spatial component enters the model (i.e., the largest $\lambda$ for which $\hat\bz_{k,\cdot}\neq\bzero$). We plot the size of each spatial component versus the value of $\lambda$ at which the spatial component enters for (a) $\tilde\ba^f_{\cdot, k}=\ba^f_{\cdot, k}$, (b) $\tilde\ba^f_{\cdot, k}=\ba^f_{\cdot, k}/\|\ba^f_{\cdot ,k}\|_2$, and (c) $\tilde\ba^f_{\cdot, k}=\ba^f_{\cdot, k}/\|\ba^f_{\cdot, k}\|_2^2$. There is a high correlation in the scatterplots in panels (a) and (b), but little correlation in (c). This motivates us to use the scaling $\tilde\ba^f_{\cdot, k}=\ba^f_{\cdot, k}/\|\ba^f_{\cdot, k}\|_2^2$ in Step 3, so that dictionary elements receive a fair shot of selection by the sparse group lasso \eqref{eq:neuronsgl}, regardless of their size.}
\label{fig:scaleA}
\end{figure}

\subsection{Sparsity of the Solution}
\label{subsec:neuronsparse}

We now consider the range of $\lambda$ for which the solution to \eqref{eq:neuronsgl} is completely sparse (i.e., $\hat\bZ=\bzero$) for a fixed value of $\alpha$. 

\begin{lemma}
\label{lem:maxlam}
For any $\alpha\in[0,1]$, the solution to \eqref{eq:neuronsgl} is completely sparse if and only if 
\begin{equation}
\lambda(1-\alpha)  \geq \left\|  \left(\left[(\tilde\bA^f)^\top\bY\right]_{k,\cdot} - \lambda\alpha\bone \right)_+ \right\|_2
\label{eq:neuronlameq}
\end{equation}
for $k=1,\ldots,{K_f}$.
\end{lemma} 
Unfortunately, when $\alpha\in (0,1)$, $\lambda$ is on both sides of the inequality in \eqref{eq:neuronlameq}. Though we can solve for $\lambda$ in \eqref{eq:neuronlameq} using a root finder when $\alpha\in (0,1)$, the following corollary provides a simple alternative.

\begin{corollary}
\label{cor:neuronmaxlam}
For any $\alpha\in(0,1)$, if 
\begin{equation}
\lambda \geq \max_{k=1,\ldots,{K_f}} \left[ \min\left( \max_{l=1,\ldots,T} \frac{\left(\left[(\tilde\bA^f)^\top\bY\right]_{k,l}\right)_+}{\alpha}, \frac{\left\|\left(\left[ (\tilde\bA^f)^\top\bY\right]_{k,\cdot} \right)_+ \right\|_2}{1-\alpha}\right)\right],\nonumber
\end{equation}
then the solution to \eqref{eq:neuronsgl} is completely sparse.
\end{corollary} 

The condition in Corollary~\ref{cor:neuronmaxlam} is sufficient, but not necessary. Proofs of Lemma~\ref{lem:maxlam} and Corollary~\ref{cor:neuronmaxlam} can be found in Appendices~\ref{app:neuronmaxlam} and \ref{app:neuronmaxlam2}, respectively. An illustration of Lemma~\ref{lem:maxlam} and Corollary~\ref{cor:neuronmaxlam} is provided in Figure~\ref{fig:maxlam}. In Section~\ref{subsec:neurontuning3}, we discuss how the results from Lemma~\ref{lem:maxlam} and Corollary~\ref{cor:neuronmaxlam} can assist in selecting the value of the tuning parameter $\lambda$.

\begin{figure}
\begin{center}
\includegraphics[width=9cm]{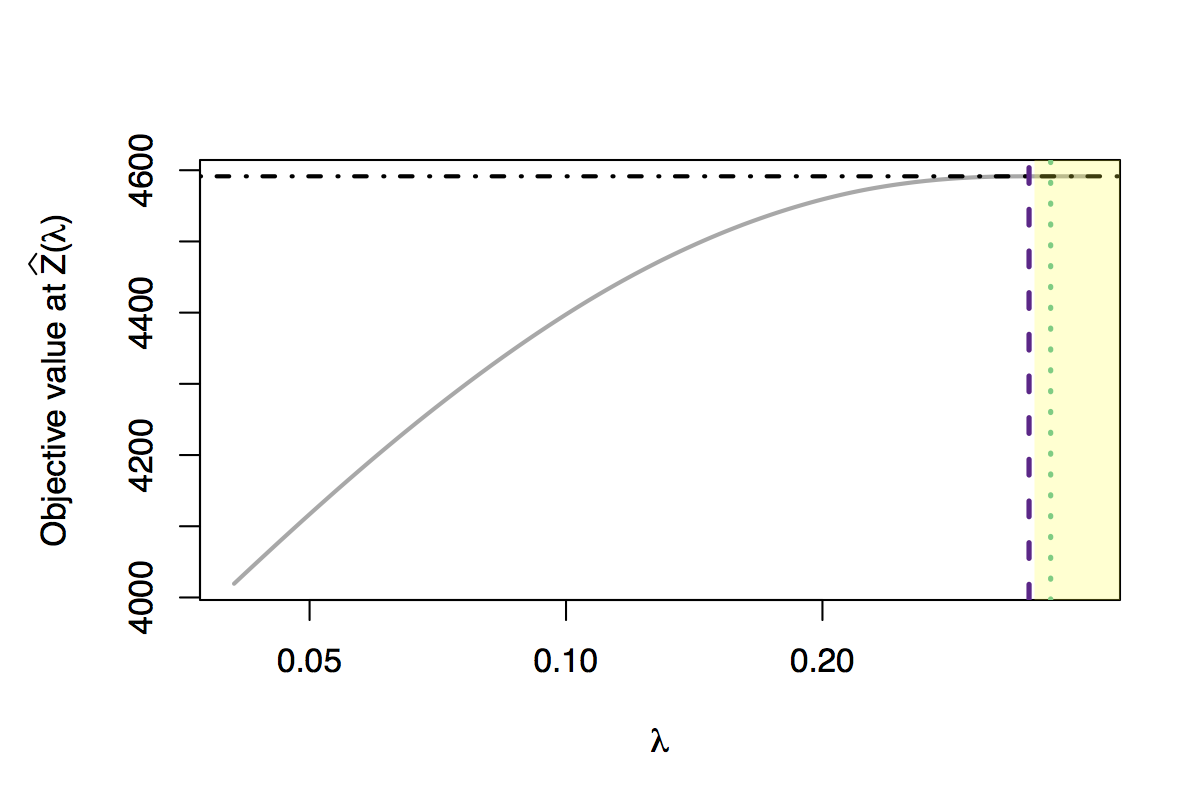}
\end{center}
\caption{We plot the value of the objective of \eqref{eq:neuronsgl} at $\hat\bZ(\lambda)$, the minimizer of \eqref{eq:neuronsgl} at $\lambda$, for a replicate of data as $\lambda$ varies. We compare two ways of finding a $\lambda$ large enough such that $\hat\bZ(\lambda)=\bzero$, which results in the objective shown as \protect\includegraphics[height=.19cm]{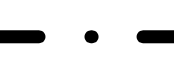}. We take $\lambda$ that satisfies Lemma~\ref{lem:maxlam} (\protect\includegraphics[height=.15cm]{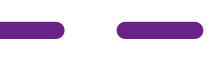}) or $\lambda$ as defined in Corollary~\ref{cor:neuronmaxlam} (\protect\includegraphics[height=.15cm]{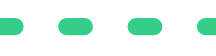}). The former (\protect\includegraphics[height=.15cm]{legend_dash1}) gives the smallest $\lambda$ such that $\hat\bZ(\lambda)=\bzero$. We can see this from the fact that the line (\protect\includegraphics[height=.15cm]{legend_dash1}) is on the boundary of the shaded box, which indicates the range of $\lambda$ for which the the objective value at $\hat\bZ(\lambda)$ equals the objective value at $\bzero$, i.e., the range of $\lambda$ for which $\hat\bZ(\lambda)=\bzero$.}
\label{fig:maxlam}
\end{figure}

\subsection{Zeroing Out of Double Neurons}
\label{subsec:neurondouble}

Some of the elements in the preliminary dictionary obtained in Step 1 may be \textit{double neurons}, i.e., elements that are a combination of two separate neurons. This occurs when two neighboring neurons are active during the same frame. In Step 2 of SCALPEL, these double neurons are unlikely to cluster with elements representing either of the individual neurons they combine, and thus there may be double neurons that remain in the filtered set of dictionary elements, $\bA^f$, used in Step 3 of SCALPEL. Fortunately, as detailed in the following lemma,  the group lasso penalty in \eqref{eq:neuronsgl} filters out these double neurons by estimating their temporal components to be the zero vector.

\begin{lemma}
\label{lem:doubleROI}
Suppose that the following conditions hold on $\bA^f\in\mathbb{R}^{P\times {K_f}}$:
\begin{enumerate}
 \item[(i)] $\ba^f_{\cdot, 3}=\ba^f_{\cdot, 1}+\ba^f_{\cdot, 2}$, 
 \item[(ii)] $(\ba^f_{\cdot, 1})^\top\ba^f_{\cdot, 2}=0$, 
 \item[(iii)] $(\ba^f_{\cdot, 1})^\top\ba^f_{\cdot, k}=0$ for $k=4,\ldots,{K_f}$, and 
 \item[(iv)] $(\ba^f_{\cdot, 2})^\top\ba^f_{\cdot, k}=0$ for $k=4,\ldots,{K_f}$.
 \end{enumerate}
 Then, define $\tilde\ba^f_{\cdot,k}\equiv\ba^f_{\cdot,k}/\| \ba^f_{\cdot,k}\|_2$, and consider solving \eqref{eq:neuronsgl} for $\bZ$ with $\alpha<1$. Then, $\hat\bz_{3, \cdot}=\bzero$.
\end{lemma}
The proof of Lemma~\ref{lem:doubleROI} is in Appendix~\ref{app:doubleROI}. Note that Lemma~\ref{lem:doubleROI} assumes that the individual elements for the neighboring neurons, $\ba^f_{\cdot,1}$ and $\ba^f_{\cdot,2}$, do not overlap at all. The group lasso penalty can also be effective at zeroing out double neurons resulting from overlapping neurons, though this depends on the amount of overlap, among other factors.


\section{Discussion}
\label{sec:neurondisc}

We have presented SCALPEL, a method for simultaneously identifying neurons from calcium imaging data and estimating their intracellular calcium concentrations. SCALPEL takes a dictionary learning approach. We segment the frames of the calcium imaging video to construct a large preliminary dictionary of potential neurons, which is then refined through the use of clustering using a novel dissimilarity metric that leverages both spatial and temporal information. The calcium concentrations of the elements of the refined dictionary are then estimated by solving a sparse group lasso with a non-negativity constraint. 

Future work could consider alternative ways of deriving a preliminary dictionary in Step 1. Currently, we perform image segmentation via thresholding with multiple quantiles. This approach assumes that active neurons will have brightness, relative to their baseline fluorescence levels, that is within the range of our image segmentation threshold values. In practice, there is evidence that some neurons have comparatively lower fluorescence following spiking, which presents a challenge for optimal identification. Though SCALPEL performed well on the one-photon calcium imaging video we considered in Section~\ref{subsec:onephoton}, other one-photon videos may have more severe background effects. If this is the case, it may be desirable to incorporate more sophisticated modeling of the background noise, like that employed in \citet{zhou2016efficient}. Additionally, in future work, we could modify Step 3 of SCALPEL to make use of a more refined model for neuron spiking, as in \citet{friedrich2015fast, vogelstein2010fast}.

 Our SCALPEL proposal is implemented in the \verb=R= package \verb=scalpel=, which is available on \verb=CRAN=. A vignette illustrating how to use the package and code to reproduce all results presented in this paper are available at  \verb=ajpete.com/software=.

\section*{Acknowledgements}

We thank Ilana Witten and Malavika Murugan at the Princeton Neuroscience Institute for providing guidance and access to the calcium imaging data set analyzed in Section~\ref{subsec:onephoton}. We also thank Michael Buice at the Allen Institute for Brain Science for his helpful comments and providing access to the calcium imaging data sets analyzed in Section~\ref{subsec:twophoton}. We also thank Liam Paninski and Pengcheng Zhou for helpful suggestions and for providing software implementing CNMF-E.


\clearpage
\appendix
\section{Supplementary Material}
\label{sec:neuronapp}

\subsection{Data Pre-Processing}
\label{app:neuronpreprocess}

To begin, we perform three pre-processing steps on the raw data. These are briefly described in Step 0 of Section~\ref{sec:neuronmethod}. First, we smooth the raw $P\times T$ data matrix spatially and temporally using a Gaussian kernel smoother with a bandwidth of 1 pixel \citep{trevor2009elements}. Second, we adjust for any bleaching effect over time. Specifically, we fit a smoothing spline with 10 degrees of freedom to the median fluorescence for each frame over time, and subtract the frame-specific smoothed median from the corresponding frame. An example of a smoothing spline fit for the example calcium imaging video is shown in Figure~\ref{fig:bleaching}.

\begin{figure}
\begin{center}
\includegraphics[width=7cm]{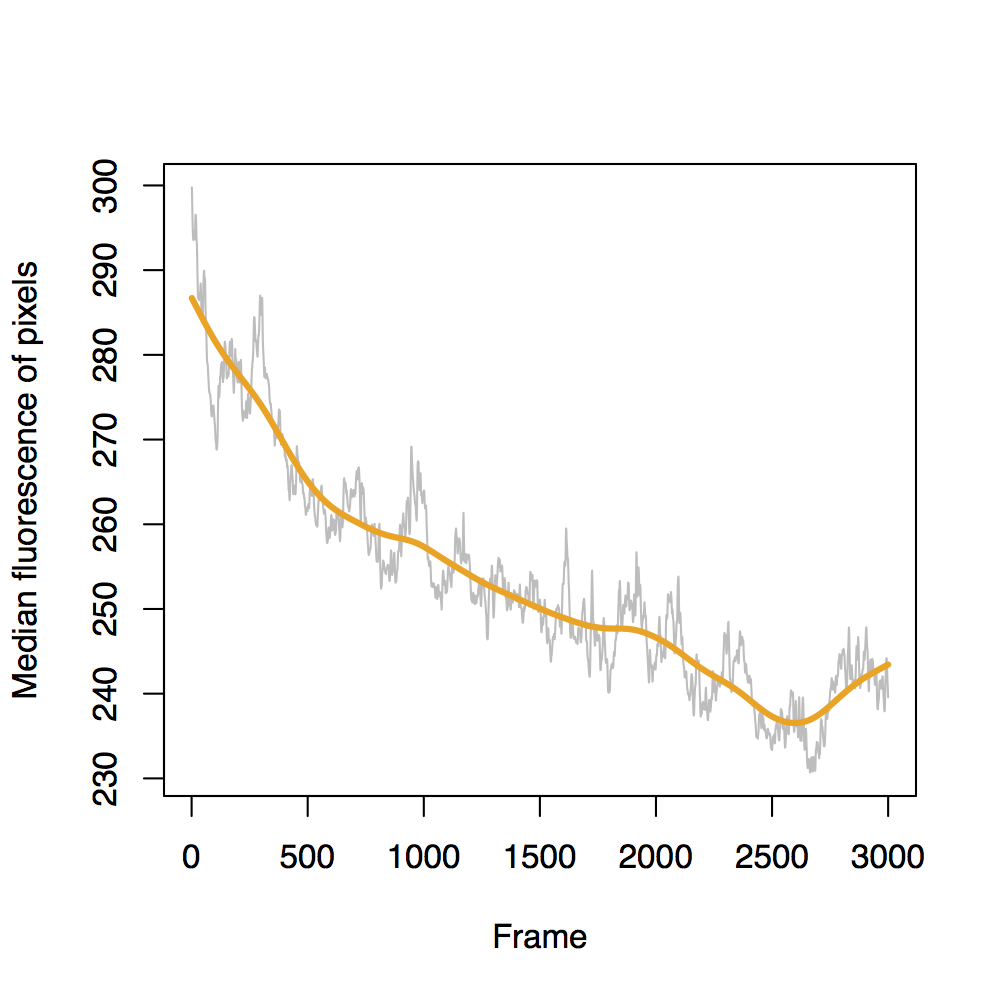}
\end{center}
\caption{For the example calcium imaging video, we plot the median fluorescence of the pixels within each frame, along with the smoothing spline fit (\protect\includegraphics[height=.13cm]{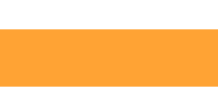}) that is used to correct for the bleaching effect.}
\label{fig:bleaching}
\end{figure}

Finally, we apply a slight variation of the often-used $\Delta f / f$ transformation \citep{grewe2010high, grienberger2012imaging, ahrens2013whole}. For the $i$th pixel in the $j$th frame, the standardized fluorescence is equal to 
 \begin{equation}
y_{i,j}\equiv \frac{y^0_{i,j} - \text{median}_{t=1,...,T}(y^0_{i,t})}{\text{median}_{t=1,...,T}(y^0_{i,t}) + \text{quantile}_{10\%}(\bY^0)},\nonumber
 \end{equation}
where $y^0_{i,j}$ is the fluorescence (after smoothing and bleaching correction) of the $i$th pixel in the $j$th frame.
This differs from the typical $\Delta f / f$ transformation in that (1) we standardize using the median across image frames instead of the mean; and (2) we add a small number to the denominator. This adjustment in the denominator prevents small fluctuations in the amount of fluorescence at pixels with very little overall fluorescence from resulting in extremely high standardized fluorescences. In Figure~\ref{fig:preprocess}, we show the resulting images after each stage of pre-processing for a sample frame.

\begin{figure}
\begin{center}
\includegraphics[width=12cm]{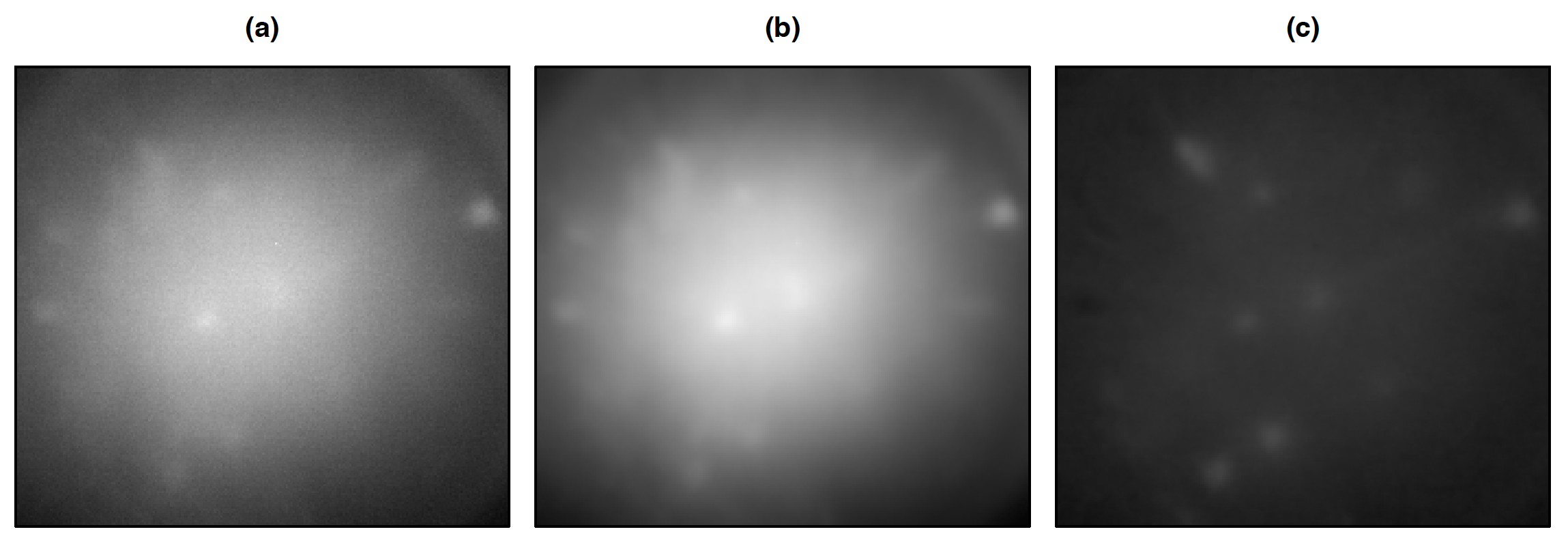}
\end{center}
\caption{In (a), we show a sample frame from the raw example calcium imaging video. In (b), we show the same frame after spatial and temporal smoothing has been done and the bleaching effect has been removed. In (c), we show the frame after the $\Delta f / f$ transformation has been performed, which is the final result of the pre-processing.}
\label{fig:preprocess}
\end{figure}

\subsection{Alternative Spatial Dissimilarity Metrics}
\label{app:neuronaltdis}

We consider three alternatives to the spatial dissimilarity proposed in \eqref{eq:neuronspatdis}: 
\begin{equation}
d^{union}_{i,j}=1-\frac{p_{i, j}}{p_{i,i}+p_{j,j}-p_{i,j}}, \text{ }d^{min}_{i,j}=1-\frac{p_{i, j}}{\min(p_{i,i}, p_{j,j})}, \text{ and }d^{max}_{i,j}=1-\frac{p_{i, j}}{\max(p_{i,i}, p_{j,j})}.
\label{eq:neuronaltdis}
\end{equation}
We refer to the dissimilarities in \eqref{eq:neuronaltdis} as the union, min, and max dissimilarities, respectively. Note that all of the dissimilarities in \eqref{eq:neuronspatdis} and \eqref{eq:neuronaltdis} take on values in $[0,1]$. 

To see why we prefer the dissimilarity measure in \eqref{eq:neuronspatdis} over those in \eqref{eq:neuronaltdis}, consider Figure~\ref{fig:spatialdis}(d), in which one component is located entirely within the other, and the two are of substantially different sizes. Min dissimilarity assigns a value of 0 in this case, despite the fact that one component is much smaller than the other. Union dissimilarity and max dissimilarity yield a very large value of 0.62, despite the fact that there is substantial overlap between the two components. In contrast, the dissimilarity in \eqref{eq:neuronspatdis} yields a modest but non-zero dissimilarity, which seems reasonable given that the two components in Figure~\ref{fig:spatialdis}(d) are similar but not identical. More generally, the dissimilarity in \eqref{eq:neuronspatdis} tends to yield reasonable values across a range of settings (Figures~\ref{fig:spatialdis}(a)-(c)). 

Often Euclidean distance, $\| \ba^0_{\cdot,i}-\ba^0_{\cdot,j}\|_2$, is used as the distance metric in clustering. Here, the squared Euclidean distance equals the number of differing pixels between components. This is problematic as a distance metric for our purposes --- it treats $a^0_{n,i}=a^0_{n,j}=0$ the same as $a^0_{n,i}=a^0_{n,j}=1$, though the latter is much more informative in terms of the components' similarity. For example, we see that the pairs of components in Figures~\ref{fig:spatialdis}(b) and (c) have similar Euclidean distances, though the components in Figure~\ref{fig:spatialdis}(b) are non-overlapping and could be arbitrarily far apart without affecting the Euclidean distance.
 
\begin{figure}
\begin{center}
\includegraphics[width=\textwidth]{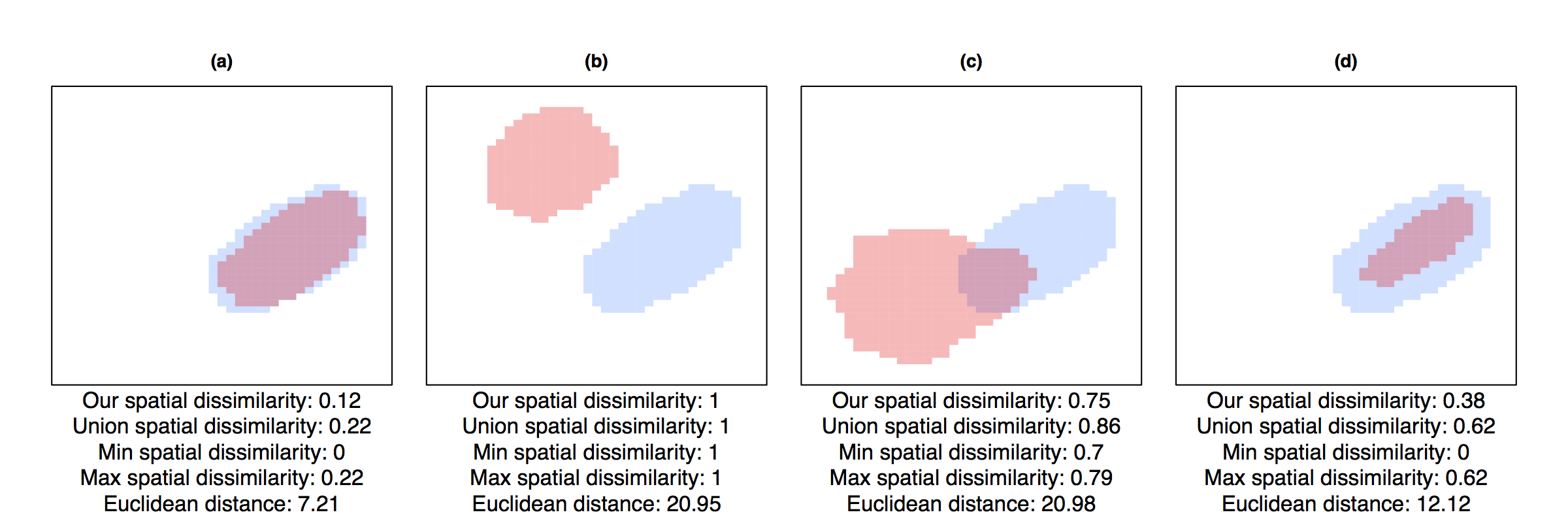}
\end{center}
\caption{We report our spatial dissimilarity metric, as defined in \eqref{eq:neuronspatdis}, and three alternative metrics, defined in \eqref{eq:neuronaltdis}, for four pairs of preliminary dictionary elements. We also report the Euclidean distance between each pair of elements.}
\label{fig:spatialdis}
\end{figure}

\subsection{Example of a Cluster in Step 2}
\label{app:neuroncluster}

To see how the preliminary dictionary element in a cluster relates to the other preliminary dictionary elements assigned to that cluster, we give an example in Figure~\ref{fig:cluster}.

\begin{figure}
\begin{center}
\includegraphics[width=\textwidth]{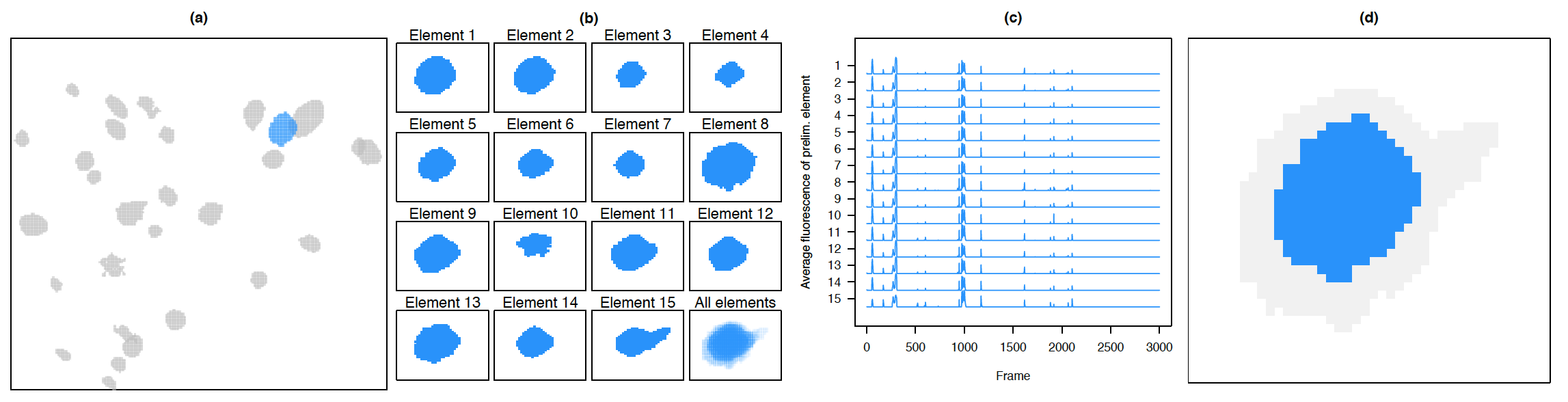}
\end{center}
\caption{We focus on a single cluster of preliminary dictionary elements from Step 2. The representative dictionary element for this cluster is highlighted in (a). The spatial maps for a random subset of the 36 preliminary dictionary elements in this cluster are shown in (b), along with a spatial map containing all 36 dictionary elements, in which the hue intensity at a given pixel indicates the number of elements containing that pixel. In (c), we plot the average thresholded fluorescence for each of the dictionary elements from (b). Finally, in (d), we show the representative element for that cluster. The gray coloring indicates the union of all preliminary dictionary elements in the cluster.}
\label{fig:cluster}
\end{figure}

\subsection{Selecting $\lambda$ for \eqref{eq:neuronsgl}}
\label{app:neuroncrossval}
To choose $\lambda$ for \eqref{eq:neuronsgl} via a validation set approach, we perform the following steps:
\begin{enumerate}
\item Obtain $\tilde\bA^f\in\mathbb{R}^{P\times {K_f}}$ by dividing the $k$th column of $\bA^f$ by $\| \ba^f_{\cdot, k}\|_2^2$.
\item Construct a training set $\mathcal{T}$ by sampling 60\% of the pixels in each overlapping group of neurons. That is, we sample 60\% of the elements in $\mathcal{M}(\mathcal{N}_1),\mathcal{M}(\mathcal{N}_2),\ldots,\mathcal{M}(\mathcal{N}_S)$. Assign the remaining pixels to the validation set, $\mathcal{V}=\{ v \in (1,\ldots,P):v\notin \mathcal{T}\}$.
\item Using Algorithm~\ref{alg:neuronfastalg}, solve \eqref{eq:neuronsgl} on the training set of pixels for a decreasing sequence of 20 $\lambda$ values, $\lambda_1,\ldots,\lambda_{20}$:
\begin{equation}
\displaystyle \hat\bZ(\lambda_i)=\underset{\bZ\in\mathbb{R}^{{K_f}\times T}, \bZ \geq 0}{\mathrm{argmin}} \quad  \frac{1}{2}  \left \| \bY_{\mathcal{T},\cdot} - \tilde\bA^f_{\mathcal{T},\cdot}\bZ\right\|_F ^2 + \lambda_i\alpha \sum_{k=1}^{{K_f}} \left\| \bz_{k,\cdot}\right\|_1+\lambda_i(1-\alpha) \sum_{k=1}^{{K_f}} \left\| \bz_{k,\cdot}\right\|_2.\nonumber
\end{equation}
\item For each $\lambda_i$, calculate the validation error,
\begin{equation}
err_{\mathcal{V}}(\lambda_{i})=\frac{1}{\vert \mathcal{V} \vert}\left\| \bY_{\mathcal{V},\cdot}^B-\tilde\bA^f_{\mathcal{V},\cdot}\hat\bZ(\lambda_i)\right\|_F^2,\nonumber
\end{equation}
where $\left[\bY_{\mathcal{V},\cdot}^B\right]_{j,k}=\begin{cases}\left[\bY_{\mathcal{V},\cdot}\right]_{j,k}&\text{if }\left[\bY_{\mathcal{V},\cdot}\right]_{j,k}>-\text{quantile}_{0.1\%}(\bY)\\0&\text{otherwise}\end{cases}$. We use a thresholded version of $\bY$ when calculating the validation error, as we only wish to have small reconstruction error on the brightest parts of the video.
Select the optimal value of $\lambda$ as
\begin{equation}
\lambda^*=\underset{\lambda_i}{\mathrm{argmax}} \left\{ \lambda_i : \frac{err_{\mathcal{V}}(\lambda_{i}) - \min_{\lambda_j} err_{\mathcal{V}}(\lambda_{j})}{ \min_{\lambda_j} err_{\mathcal{V}}(\lambda_{j})} \leq 0.05 \right\};\nonumber
\end{equation}
this is the largest value of $\lambda$ that results in a validation error within 5\% of the minimum validation error achieved by any value of $\lambda$ considered.
\item Solve \eqref{eq:neuronsgl} on all pixels:
\begin{equation}
\displaystyle \underset{\bZ\in\mathbb{R}^{{K_f}\times T}, \bZ \geq 0}{\mathrm{minimize}} \quad  \frac{1}{2}  \left \| \bY - \tilde\bA^f\bZ\right\|_F ^2 + \frac{\lambda^*}{\vert \mathcal{T}\vert/P}\alpha \sum_{k=1}^{{K_f}} \left\| \bz_{k,\cdot}\right\|_1+\frac{\lambda^*}{\vert \mathcal{T}\vert/P}(1-\alpha) \sum_{k=1}^{{K_f}} \left\| \bz_{k,\cdot}\right\|_2,\nonumber
\end{equation}
where we have scaled the tuning parameter by the percent of pixels in the training set, in order to account for the fact that the sum of squared errors in the loss function is not scaled by the number of pixels.
\end{enumerate}

This process can be done separately for each group of overlapping neurons $\mathcal{N}_1,\ldots,\mathcal{N}_S$ in order to select a different value of $\lambda$ for each group, or for all groups at once to select a single value of $\lambda$. By following steps similar to those described above, $\lambda$ can alternatively be selected via cross-validation.  

\subsection{Filtering Dictionary Elements Prior to Fitting the Sparse Group Lasso}
\label{app:minclusters}

Optionally, at the beginning of Step 3, the elements in the refined dictionary from Step 2 may be filtered prior to fitting the sparse group lasso in  \eqref{eq:neuronsgl}. One way to filter the dictionary elements is on the basis of the number of members in the clusters. That is, we can discard any dictionary elements representing clusters containing fewer than some minimum number of members. This  number  should be chosen based on the goals of the analysis.  If we retain all clusters, regardless of size, then we may include some non-neuronal dictionary elements in the sparse group lasso problem. In contrast, if we discard small clusters (e.g., those with fewer than 5 members), then we may erroneously filter out some true neurons. 

As an alternative to filtering dictionary elements on the basis of cluster size, we can instead manually inspect each dictionary element, by examining the frames of the video from which it was derived. In Figure~\ref{fig:tradeoff}, we compare the results obtained by manually filtering each dictionary element, versus simply filtering dictionary elements on the basis of cluster size, on the example video considered in Sections~\ref{sec:neuronmethod} and~\ref{subsec:onephoton}. We find that the two approaches give similar results.

\begin{figure}
\begin{center}
\includegraphics[width=2.5in]{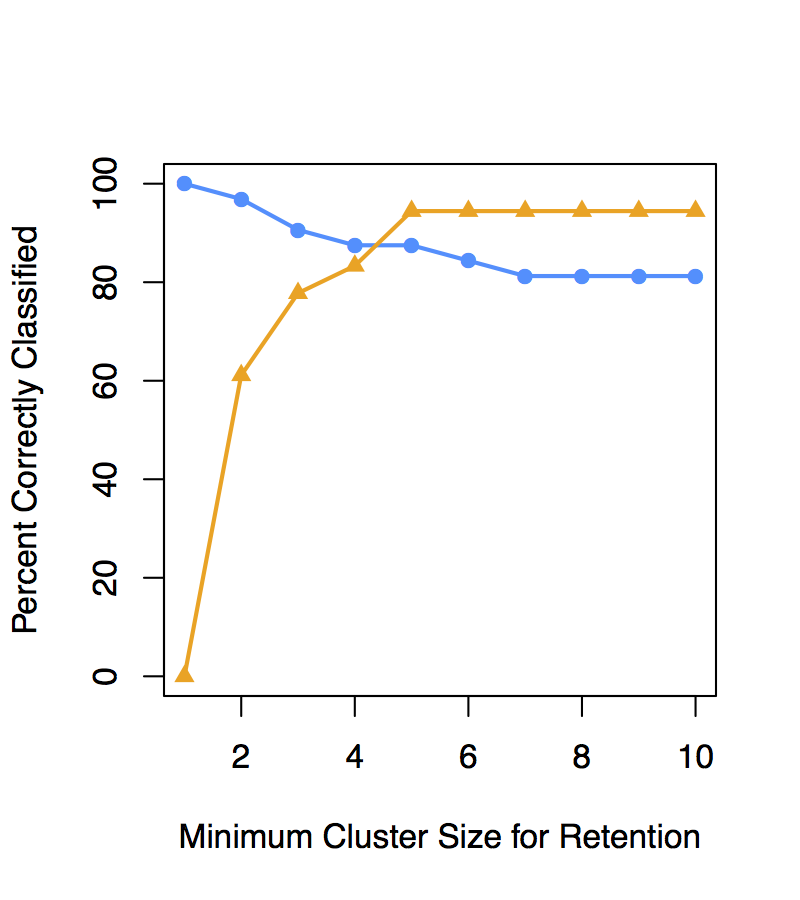}
\end{center}
\caption{On the example video considered in Sections~\ref{sec:neuronmethod} and~\ref{subsec:onephoton}, we manually inspected each dictionary element resulting from Step 2 of SCALPEL, by examining the frames from which each dictionary element was derived. Based on this manual inspection, we classified each of the 50 dictionary elements as a ``neuron" or a ``non-neuron". Next, we considered whether simply filtering each dictionary element based on the number of elements in its cluster (as described at the beginning of Step 3 of SCALPEL) would accurately distinguish between ``neurons" and ``non-neurons". In the figure, the y-axis shows the percentage of ``neurons" that would remain after filtering (\protect\includegraphics[height=.23cm]{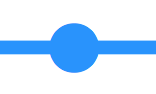}), and the percentage of ``non-neurons" that would be eliminated via filtering (\protect\includegraphics[height=.23cm]{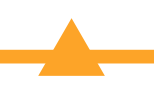}), as a function of the filtering threshold (shown on the x-axis). We find that in this video, a careful manual analysis of each dictionary element yields very similar results to simply filtering each dictionary element based on the number of elements in its cluster.}
\label{fig:tradeoff}
\end{figure}

\subsection{Proof of Lemma~\ref{lem:neurononecomp}}
\label{app:neurononecomp}

We first prove a result that we will use later.

\begin{lemma}
\label{lem:nonnegsoln}
The solution to $\underset{\bbeta\geq\bzero}{\mathrm{minimize}} \quad\frac{1}{2}\left\|\by-\bbeta\right\|_2^2+\lambda\left\|\bbeta\right\|_2$ is $\hat\bbeta=\left( 1-\frac{\lambda}{\left\| \left(\by \right)_+\right\|_2}\right)_+\left( \by\right)_+$.
\end{lemma}

\begin{proof}
Let $\hat\bbeta=\underset{\bbeta\geq\bzero}{\mathrm{argmin}}\;\frac{1}{2}\left\|\by-\bbeta\right\|_2^2+\lambda\left\|\bbeta\right\|_2$ and $\mathcal{C}=\{ i: y_i\geq 0\}$. First, we show $\hat\bbeta_{-\mathcal{C}}=\bzero$. In anticipation of contradiction, assume there exists $j$ such that $j\notin \mathcal{C}$ and $\hat\beta_j>0$. Define $\tilde\bbeta$ as $\tilde\beta_i=\begin{cases}\hat\beta_i &\text{if }i\neq j\\0&\text{if }i=j \end{cases}.$ Then
\begin{equation}
\frac{1}{2}\left\|\by-\tilde\bbeta\right\|_2^2+\lambda\left\|\tilde\bbeta\right\|_2 < \frac{1}{2}\left\|\by-\hat\bbeta\right\|_2^2+\lambda\left\|\hat\bbeta\right\|_2.\nonumber
\end{equation}
This is a contradiction, so we conclude that $\hat\beta_i=0$ for all $i\notin\mathcal{C}$. It remains to solve
\begin{equation}
\label{eq:neuronpos}
\underset{\bbeta_{\mathcal{C}}\geq\bzero}{\mathrm{minimize}} \quad\frac{1}{2}\left\|\by_{\mathcal{C}}-\bbeta_{\mathcal{C}}\right\|_2^2+\lambda\left\|\bbeta_{\mathcal{C}}\right\|_2.
\end{equation}
By a result in Section 3.1 of \citet{simon2013sparse}, the solution to \eqref{eq:neuronpos} without the non-negativity constraint on $\bbeta_{\mathcal{C}}$ is $\left( 1-\frac{\lambda}{\left\| \by_{\mathcal{C}}\right\|_2}\right)_+\by_{\mathcal{C}}$, which has all non-negative elements. Therefore, it is also the solution to \eqref{eq:neuronpos}.
\end{proof}

We now proceed to prove Lemma~\ref{lem:neurononecomp}.

\begin{proof}
Our goal is to solve
\begin{equation}
\label{eq:neurononecompapp}
\displaystyle \underset{\bz\in\mathbb{R}^{T},\bz\geq \bzero}{\mathrm{minimize}} \quad  \frac{1}{2}  \left \| \bY - \tilde\ba^f\bz^\top\right\|_F ^2 + \lambda\alpha\bone^\top \bz+\lambda(1-\alpha)  \left\| \bz\right\|_2.
\end{equation}
Note that solving \eqref{eq:neurononecompapp} is equivalent to solving \eqref{eq:neurononecomp}, as $\|\bz\|_1=\bone^\top\bz$ when $\bz\geq\bzero$. By algebraic manipulation, we can show that
\begin{equation}
 \left\| \bY - \tilde\ba^f\bz^\top\right\|_F^2+ \lambda\alpha\bone^\top\bz=  \left\| \frac{\bY^\top\tilde\ba^f- \lambda\alpha\bone}{\sqrt{(\tilde\ba^f)^\top\tilde\ba^f}} - \sqrt{(\tilde\ba^f)^\top\tilde\ba^f}\bz\right\|_2^2+C,\nonumber
 \end{equation}
where $C$ is a constant that does not depend on $\bz$. Therefore, the solution to \eqref{eq:neurononecompapp} is the same as the solution to
\begin{equation}
\label{eq:neuronsoln}
\displaystyle \underset{\bz\geq \bzero}{\mathrm{minimize}} \quad  \frac{1}{2}  \left\| \frac{\bY^\top\tilde\ba^f - \lambda\alpha\bone}{(\tilde\ba^f)^\top\tilde\ba^f} -\bz\right\|_2^2+\frac{\lambda(1-\alpha)}{(\tilde\ba^f)^\top\tilde\ba^f}  \left\| \bz\right\|_2.
\end{equation}
We solve \eqref{eq:neuronsoln} by applying Lemma~\ref{lem:nonnegsoln}.
\end{proof}

\subsection{Proof of Lemma~\ref{lem:neurondecomp}}
\label{app:neurondecomp}

\begin{proof}
The result follows simply from observing that
\begin{align}
\left\|\bY - \tilde{\bA^f} \bZ \right\|_F^2 &= \sum_{s=1}^S \left\|\bY_{\mathcal{M}(\mathcal{N}_s),\cdot} - \tilde{\bA^f}_{\mathcal{M}(\mathcal{N}_s),\cdot} \bZ \right\|_F^2 \nonumber\\
&= \sum_{s=1}^S \left\|\bY_{\mathcal{M}(\mathcal{N}_s),\cdot} - \sum_{s'=1}^S \tilde{\bA^f}_{\mathcal{M}(\mathcal{N}_s),\mathcal{N}_{s'}} \bZ_{\mathcal{N}_{s'}, \cdot} \right\|_F^2\nonumber\\
&=    \sum_{s=1}^S \left\|\bY_{\mathcal{M}(\mathcal{N}_s),\cdot} - \tilde{\bA^f}_{\mathcal{M}(\mathcal{N}_s),\mathcal{N}_s} \bZ_{\mathcal{N}_s, \cdot} \right\|_F^2.\nonumber
\end{align}
The last equality follows from the condition of the lemma, which guarantees that $\tilde{\bA^f}_{\mathcal{M}(\mathcal{N}_s),\mathcal{N}_{s'}}=0$ for all $s \neq s'$.
\end{proof}

\subsection{Details of Step 2(b) of Algorithm~\ref{alg:neuronfastalg}}
\label{app:neuronappggd}
Note that minimizing the objective in \eqref{eq:neuronsgllittle} subject to $\bZ_{\mathcal{N}_{s}, \cdot}\geq\bzero$ is equivalent to minimizing
\begin{equation}
\label{eq:withnolasso}
\displaystyle \quad  \frac{1}{2}  \left\| \bY_{\mathcal{M}(\mathcal{N}_s), \cdot} - \tilde\bA^f_{\mathcal{M}(\mathcal{N}_s),\mathcal{N}_{s}}\bZ_{\mathcal{N}_{s}, \cdot}\right\|_F^2+ \lambda\alpha  \bone^\top\bZ_{\mathcal{N}_s, \cdot}\bone+\lambda(1-\alpha) \sum_{n\in\mathcal{N}_s} \left\| \bz_{n,\cdot}\right\|_2
\end{equation}
subject to $\bZ_{\mathcal{N}_{s}, \cdot}\geq\bzero$, since $ \bone^\top\bZ_{\mathcal{N}_s, \cdot}\bone=\sum_{n\in\mathcal{N}_s} \left\| \bz_{n,\cdot}\right\|_1$ when $\bZ_{\mathcal{N}_{s}, \cdot}\geq\bzero$. 

Let $f\left(\bZ_{\mathcal{N}_s, \cdot}\right)= \frac{1}{2}  \left \| \bY_{\mathcal{M}(\mathcal{N}_s),\cdot} - \tilde\bA^f_{\mathcal{M}(\mathcal{N}_s),\mathcal{N}_s}\bZ_{\mathcal{N}_s, \cdot}\right\|_F ^2 + \lambda\alpha \bone^\top\bZ_{\mathcal{N}_s, \cdot}\bone$, which is the differentiable part of \eqref{eq:withnolasso}, and let $g\left(\bZ_{\mathcal{N}_s, \cdot}\right)=\lambda(1-\alpha)\sum_{n\in\mathcal{N}_s}\left\| \bz_{n,\cdot}\right\|_2$, the non-differentiable part.

Generalized gradient descent \citep{beck2009fast, parikh2014proximal} is a majorization-minimization scheme. First, we find a quadratic approximation to $f\left(\bZ_{\mathcal{N}_s, \cdot}\right)$ centered at our previous estimate for $\bZ_{\mathcal{N}_s, \cdot}$, $\bZ_{\mathcal{N}_s, \cdot}^0$, that majorizes $f\left(\bZ_{\mathcal{N}_s, \cdot}\right)$. That is,
\begin{equation}
f\left(\bZ_{\mathcal{N}_s, \cdot}\right)\leq f\left(\bZ_{\mathcal{N}_s ,\cdot}^0\right) + \Tr\left[\left(\bZ_{\mathcal{N}_s, \cdot} - \bZ_{\mathcal{N}_s, \cdot}^0\right)^\top\nabla f\left(\bZ_{\mathcal{N}_s, \cdot}^0\right)\right]+\frac{1}{2t}\left\| \bZ_{\mathcal{N}_s, \cdot} - \bZ_{\mathcal{N}_s, \cdot}^0\right\|_F^2,\nonumber
\end{equation}
where $t$ is the step size such that $\nabla^2 f(\cdot)\preceq \frac{1}{t}\bI$. After completing the square, we can see that minimizing the quadratic approximation to $f\left(\bZ_{\mathcal{N}_s, \cdot}\right)$ gives the same solution as solving
\begin{equation}
\underset{\bZ_{\mathcal{N}_s, \cdot}}{\mathrm{minimize}} \quad\frac{1}{2t}\left\| \bZ_{\mathcal{N}_s, \cdot} - \left( \bZ_{\mathcal{N}_s, \cdot}^0- t\nabla f\left(\bZ_{\mathcal{N}_s, \cdot}^0\right)\right) \right\|_F^2.\nonumber
\end{equation}
Thus we perform this minimization with $g\left(\bZ_{\mathcal{N}_s, \cdot}\right)$ added to the objective function, which gives the proximal problem
\begin{equation}
\label{eq:neuronprox}
\underset{\bZ_{\mathcal{N}_s ,\cdot}\geq\bzero}{\mathrm{minimize}} \quad\frac{1}{2}\left\| \bZ_{\mathcal{N}_s, \cdot} -\tilde\bY_{\mathcal{N}_s,\cdot} \right\|_F^2 + \lambda(1-\alpha)t\sum_{n\in\mathcal{N}_s} \left\| \bz_{n,\cdot}\right\|_2,
\end{equation}
where $\tilde\bY_{\mathcal{N}_s,\cdot}=\bZ_{\mathcal{N}_s, \cdot}^0- t\left( -(\tilde\bA^f_{\mathcal{M}(\mathcal{N}_s),\mathcal{N}_s})^\top\left(\bY_{\mathcal{M}(\mathcal{N}_s),\cdot}-\tilde\bA^f_{\mathcal{M}(\mathcal{N}_s),\mathcal{N}_s}\bZ_{\mathcal{N}_s, \cdot}^0\right) + \lambda\alpha\bone\bone^\top\right)$. The minimization in \eqref{eq:neuronprox} is separable in $\bz_{n,\cdot}$, so for $n\in\mathcal{N}_s$, we solve
\begin{equation}
\label{eq:neuronproxsep}
\underset{\bz_{n,\cdot}\geq\bzero}{\mathrm{minimize}} \quad\frac{1}{2}\left\| \bz_{n,\cdot} -\tilde\by_{n,\cdot} \right\|_2^2 + \lambda(1-\alpha)t\left\| \bz_{n,\cdot}\right\|_2.
\end{equation}
By Lemma~\ref{lem:nonnegsoln} in Appendix~\ref{app:neurononecomp}, the solution to \eqref{eq:neuronproxsep} is $\hat\bz_{n,\cdot}=\left( 1-\frac{\lambda(1-\alpha)t}{\left\| \left(\tilde\by_{n,\cdot}\right)_+\right\|_2}\right)_+\left(\tilde\by_{n,\cdot}\right)_+$.

It only remains to derive a suitable step size $t$ so that $\nabla^2 f(\cdot)=(\tilde\bA^f_{\mathcal{M}(\mathcal{N}_s),\mathcal{N}_s})^\top\tilde\bA^f_{\mathcal{M}(\mathcal{N}_s),\mathcal{N}_s}\preceq \frac{1}{t}\bI$. A sufficient condition for $\frac{1}{t}\bI-(\tilde\bA^f_{\mathcal{M}(\mathcal{N}_s),\mathcal{N}_s})^\top\tilde\bA^f_{\mathcal{M}(\mathcal{N}_s),\mathcal{N}_s}$ to be positive semi-definite is that $\frac{1}{t}\bI-(\tilde\bA^f_{\mathcal{M}(\mathcal{N}_s),\mathcal{N}_s})^\top\tilde\bA^f_{\mathcal{M}(\mathcal{N}_s),\mathcal{N}_s}$ be diagonally dominant. That is,
\begin{equation}
\frac{1}{t} -(\tilde\ba^f_{\mathcal{M}(\mathcal{N}_s),n})^\top\tilde\ba^f_{\mathcal{M}(\mathcal{N}_s),n}\geq \sum_{j\in\mathcal{N}_s,j\neq n} (\tilde\ba^f_{\mathcal{M}(\mathcal{N}_s),j})^\top\tilde\ba^f_{\mathcal{M}(\mathcal{N}_s),n}\nonumber
\end{equation}
for all $n\in\mathcal{N}_s$.
Thus we choose $t=(\max_{n\in\mathcal{N}_s} \sum_{j\in\mathcal{N}_s} (\tilde\ba^f_{\mathcal{M}(\mathcal{N}_s),j})^\top\tilde\ba^f_{\mathcal{M}(\mathcal{N}_s),n})^{-1}$.

\subsection{Proof of Lemma~\ref{lem:Ascale}}
\label{app:neuronAscale}

\begin{proof}
Since $\ba^f_{\cdot ,1}$ does not overlap any other spatial components (i.e., $(\ba^f_{\cdot, 1})^\top\ba^f_{\cdot, k}=0$ for $k=2,\ldots,{K_f}$), we know by the results in Lemmas~\ref{lem:neurononecomp} and \ref{lem:neurondecomp} that
\begin{equation}
\hat\bz_{1, \cdot} = \left(1 - \frac{\lambda(1-\alpha)}{\| ( \bY^\top\tilde\ba^f_{\cdot, 1}-\lambda\alpha\bone  )_+\|_2} \right)_+  \left( \frac{\bY^\top \tilde\ba^f_{\cdot, 1}-\lambda\alpha\bone }{(\tilde\ba^f_{\cdot, 1})^\top \tilde\ba^f_{\cdot, 1}}\right)_+.\nonumber
\end{equation}
Note that $\bY^\top\tilde\ba^f_{\cdot, 1}=(\bZ^*)^\top(\bA^f)^\top\ba^f_{\cdot, 1}/\|\ba^f_{\cdot, 1}\|_2^2$, so $\bY^\top\tilde\ba^f_{\cdot, 1}=\bz_{1, \cdot}^*$ since $(\ba^f_{\cdot, 1})^\top\ba^f_{\cdot, k}=0$ for $k=2,\ldots,{K_f}$. Thus $\hat\bz_{1, \cdot} =\bzero$ if and only if $\lambda(1-\alpha)\geq \| (\bz_{1, \cdot}^* - \lambda\alpha\bone)_+\|_2$. Similarly, $\hat\bz_{2, \cdot} =\bzero$ if and only if $\lambda(1-\alpha)\geq \| (\bz_{2, \cdot}^* - \lambda\alpha\bone)_+\|_2$. Thus, it remains to show that $\| (\bz_{1, \cdot}^* - \lambda\alpha\bone)_+\|_2=\| (\bz_{2, \cdot}^* - \lambda\alpha\bone)_+\|_2$. Using the properties of permutation matrices that $\bone=\bP\bone$ and $\bP^\top\bP=\bI$, we see 
\begin{align}
\| (\bz_{1, \cdot}^* - \lambda\alpha\bone)_+\|_2 &=\| (\bP\bz_{2, \cdot}^* - \lambda\alpha\bP\bone)_+\|_2 \nonumber\\
&=\| \bP(\bz_{2, \cdot}^*- \lambda\alpha\bone)_+\|_2 \nonumber\\
&=\| (\bz_{2, \cdot}^*- \lambda\alpha\bone)_+\|_2 \nonumber,
\end{align}
and therefore, $\hat\bz_{1, \cdot} =\bzero$ if and only if $\hat\bz_{2, \cdot} =\bzero$ .
\end{proof}

\subsection{Proof of Lemma \ref{lem:maxlam}}
\label{app:neuronmaxlam}

\begin{proof}
Recall that solving  \eqref{eq:neuronsgl} gives the same solution as solving \eqref{eq:neuronreducesgl}. Thus we focus on deriving a condition on $\lambda$ that guarantees that $\hat\bZ_{\mathcal{N}_s,\cdot}$, the solution to \eqref{eq:neuronreducesgl}, equals zero for $s=1,\ldots,S$.
If $\vert \mathcal{N}_s\vert =1$, we see from \eqref{eq:neuronsinglezsolnalg} that $\hat\bz_{\mathcal{N}_s,\cdot}=\bzero$ if and only if 
\begin{equation}
\label{eq:neuronconditionsingle}
\lambda(1-\alpha)\geq \left\|\left( (\bY_{\mathcal{M}(\mathcal{N}_s),\cdot})^\top\tilde\ba^f_{\mathcal{M}(\mathcal{N}_s),\mathcal{N}_s} -\lambda\alpha\bone \right)_+\right\|_2.
\end{equation}
Recall that if $\vert\mathcal{N}_s\vert>1$, we iteratively solve for $\hat\bZ_{\mathcal{N}_s,\cdot}$ using Step 2(b) of Algorithm~\ref{alg:neuronfastalg}. We initialize at the sparse solution $\bZ_{\mathcal{N}_s,\cdot}^{(0)}=\bzero$ and thus for $n\in\mathcal{N}_s$
\begin{equation}
 \bz^{(1)}_{n,\cdot}=\left(1-\frac{\lambda(1-\alpha)t}{\left\| \left( \tilde\by_{n,\cdot}\right)_+\right\|_2} \right)_+ \left( \tilde\by_{n,\cdot}\right)_+,\nonumber
 \end{equation} 
where $\tilde\bY_{\mathcal{N}_s,\cdot}=t(\tilde\bA^f_{\mathcal{M}(\mathcal{N}_s),\mathcal{N}_s})^\top\bY_{\mathcal{M}(\mathcal{N}_s),\cdot}-t\lambda\alpha\bone\bone^\top$. We will have $\hat\bZ_{\mathcal{N}_s,\cdot}=\bzero$ if $\bz_{n,\cdot}^{(1)}=\bzero$ for all $n\in\mathcal{N}_s$. Note that $\bz^{(1)}_{n,\cdot}=\bzero$ if 
\begin{equation}
\label{eq:neuroncondition}
\lambda(1-\alpha)t  \geq \left\|  \left(t\left[(\tilde\bA^f_{\mathcal{M}(\mathcal{N}_s),\mathcal{N}_s})^\top\bY_{\mathcal{M}(\mathcal{N}_s),\cdot}\right]_{n,\cdot} - t\lambda\alpha\bone \right)_+ \right\|_2.
\end{equation}
By algebraic manipulation, the sparsity conditions given in \eqref{eq:neuronconditionsingle} and \eqref{eq:neuroncondition} can be shown to be equivalent to the condition given in Lemma~\ref{lem:maxlam}. Alternatively, this lemma's result also follows from inspection of the optimality condition for \eqref{eq:neuronsgl}. 
\end{proof}

\subsection{Proof of Corollary~\ref{cor:neuronmaxlam}}
\label{app:neuronmaxlam2}

\begin{proof}
The sufficient condition given in Corollary~\ref{cor:neuronmaxlam} follows from noting that \eqref{eq:neuroncondition} is satisfied if
$\lambda(1-\alpha)  \geq \left\|  \left(\left[(\tilde\bA^f)^\top\bY\right]_{k,\cdot} \right)_+ \right\|_2$
or if
$\lambda\alpha\geq\left(\left[(\tilde\bA^f)^\top\bY\right]_{k,l}\right)_+$
for $l=1,\ldots,T$. Thus, when at least one of these two conditions is satisfied for all $k=1,\ldots,{K_f}$, then the solution to \eqref{eq:neuronsgl} will be sparse.
\end{proof}

\subsection{Proof of Lemma \ref{lem:doubleROI}}
\label{app:doubleROI}
\begin{proof}
Let $\hat\bZ$ be the solution to \eqref{eq:neuronsgl}. In anticipation of contradiction, assume there exists $j\in\{1,\ldots,T\}$ such that $\hat\bz_{3,j}>0$. Define $\tilde\bZ$ as $\tilde\bz_{1,\cdot}=\hat\bz_{1,\cdot} + \left(\frac{\bone^\top\ba^f_{\cdot,1}}{\bone^\top(\ba^f_{\cdot,1}+\ba^f_{\cdot,2})} \right) \hat\bz_{3,\cdot}$, $\tilde\bz_{2,\cdot}=\hat\bz_{2,\cdot} + \left(\frac{\bone^\top\ba^f_{\cdot,2}}{\bone^\top(\ba^f_{\cdot,1}+\ba^f_{\cdot,2})} \right) \hat\bz_{3,\cdot}$, $\tilde\bz_{3,\cdot}=\bzero$, and $\tilde\bz_{k,\cdot}=\hat\bz_{k,\cdot}$ for $k=4,\ldots,K^f$. Let $\text{obj}(\bZ)$ be the value of the objective function of \eqref{eq:neuronsgl} at $\bZ$ for some fixed $\lambda$ and $\alpha$. We have
 \begin{align}
\text{obj}(\tilde\bZ)-\text{obj}(\hat\bZ) \nonumber&= \lambda(1-\alpha)\sum_{k=1}^3 \left(\| \tilde\bz_{k,\cdot}\|_2 - \| \hat\bz_{k,\cdot}\|_2\right)\\\nonumber
&=\lambda(1-\alpha)[ \|\hat\bz_{1,\cdot} + (\bone^\top\ba^f_{\cdot,1})/(\bone^\top(\ba^f_{\cdot,1}+\ba^f_{\cdot,2}))\hat\bz_{3,\cdot}\|_2 \nonumber\\ & \qquad + \|\hat\bz_{2,\cdot} + (\bone^\top\ba^f_{\cdot,2})/(\bone^\top(\ba^f_{\cdot,1}+\ba^f_{\cdot,2}))\hat\bz_{3,\cdot}\|_2   - \left( \|\hat\bz_{1,\cdot}\|_2 + \|\hat\bz_{2,\cdot}\|_2 + \|\hat\bz_{3,\cdot}\|_2\right)]\nonumber\\
&< \lambda(1-\alpha)[ \|\hat\bz_{1,\cdot} \|_2+ (\bone^\top\ba^f_{\cdot,1})/(\bone^\top(\ba^f_{\cdot,1}+\ba^f_{\cdot,2}))\|\hat\bz_{3,\cdot}\|_2 + \|\hat\bz_{2,\cdot}\|_2 \nonumber\\ & \qquad+ (\bone^\top\ba^f_{\cdot,2})/(\bone^\top(\ba^f_{\cdot,1}+\ba^f_{\cdot,2}))\|\hat\bz_{3,\cdot}\|_2   - \left( \|\hat\bz_{1,\cdot}\|_2 + \|\hat\bz_{2,\cdot}\|_2 + \|\hat\bz_{3,\cdot}\|_2\right)]\nonumber\\
&=0.  \nonumber
\end{align}
This is a contradiction, so we conclude $\hat\bz_{3,\cdot}=\bzero$.
\end{proof}


\bibliography{reference}
\bibliographystyle{plainnat}

\end{document}